\documentclass[11pt]{article}
\usepackage[letterpaper, margin=1in]{geometry}

\usepackage{setspace}
\usepackage{tcolorbox}
\usepackage{multirow}
\usepackage{refcount}
\usepackage{picture}
\newcounter{tcbbox}

\newcommand{\tcbref}[1]{Outline~\ref{#1}}

\usepackage{amsmath, amssymb} 
\usepackage{graphicx} 
\usepackage{cite} 
\usepackage{url} 
\usepackage{verbatim}

\title{Sublinear-Time Sampling of Spanning Trees in the Congested Clique}
\author{Sriram V. Pemmaraju \\ University of Iowa \\ sriram-pemmaraju@uiowa.edu \and Sourya Roy \\ University of Iowa \\ sourya-roy@uiowa.edu \and Joshua Z. Sobel \\ University of Iowa \\ joshua-sobel@uiowa.edu}
\date{}

\usepackage{xcolor}
\usepackage{algorithm}
\usepackage{amsmath}
\usepackage{amsthm}
\usepackage{amssymb}
\usepackage[T1]{fontenc}
\usepackage[colorlinks=true,allcolors=blue]{hyperref}
\usepackage{algpseudocodex}

\newtheorem{lemma}{Lemma}
\newtheorem{corollary}{Corollary}

\newtheorem{definition}{Definition}
\newtheorem{fact}{Fact}

\newcommand{\cc}{{\scshape CongestedClique}}
\newcommand{\congest}{{\scshape Congest}}
\newcommand{\pln}{O(\mbox{poly} \log n)}
\newcommand{\dbling}{\textsc{Doubling} }
\newcommand{\id}{\texttt{Id} }
\usepackage{svg}
\usepackage{todonotes}
\definecolor{lgray}{rgb}{0.01,0.199,0.1}

\newcommand{\E}{\mathbb{E}}

\newcommand{\M}{\mathbb{M}}

\newcommand{\calH}{\mathcal{H}}

\newcommand{\sfM}{{\sf M}}

\newcommand{\chk}{{\sf CheckTruncationPoint}}

\newcommand{\adj}{\mathbf{P}}

\newcommand{\old}{{\sf Old}}
\newcommand{\sct}{\mathbf{Q}}
\newcommand{\scc}{\mathbf{S}}
\newcommand{\schur}{\textsc{Schur}}
\newcommand{\shortcut}{\textsc{ShortCut}}
\newcommand{\Ct}{\textsf{Count}}

\usepackage{cleveref}
\usepackage{thm-restate}

\begin{document}
\maketitle

\begin{abstract}
We present the first sublinear-in-$n$ round algorithm for sampling an approximately uniform spanning tree of an $n$-vertex graph in the \cc{} model of distributed computing.  In particular, our algorithm requires $\Tilde{O}(n^{0.657})$ rounds for sampling a spanning tree within total variation distance $1/n^c$, for arbitrary constant $c > 0$, from the uniform distribution.  More precisely, our algorithm requires $\Tilde{O}(n^{1/2 + \alpha})$ rounds, where $O(n^\alpha)$ is the running time of matrix multiplication in the \cc{} model (currently $\alpha = 1 - 2/\omega = 0.157$, where $\omega$ is the sequential matrix multiplication time exponent).
We can adapt our algorithm to give exact rather than approximate samples, but with a larger, though still $o(n)$, runtime of $\Tilde{O}(n^{2/3+\alpha}) = O(n^{.824})$.  

In a remarkable result, Aldous (SIDM 1990) and Broder (FOCS 1989) showed that the first visit edge to each vertex, excluding the start vertex, during a random walk forms a uniformly chosen spanning tree of the underlying graph. The biggest challenge with implementing the Aldous-Broder algorithm in a distributed setting is that it requires a very long random walk, i.e., a walk that covers the graph and might have to be $\Theta(mn)$ steps long for an $n$-vertex, $m$-edge graph, in the worst case.
The most common technique for constructing random walks in ``all-to-all'' communication models such as \cc{} and MPC are bottom-up techniques which start with many short walks and repeatedly stitch these together to build longer walks. 
These techniques fail to efficiently build the long random walk required by the Aldous-Broder algorithm.
Our algorithm is a significant departure from known techniques, featuring a top-down walk filling approach paired with  Schur complement graphs for walk shortcutting.  To make this idea work in the \cc{} model, we present a novel compressed random walk reconstruction algorithm, based on randomly sampling a weighted perfect matching.

In addition, we show how to take somewhat shorter random walks even more efficiently in the \cc{} model, obtaining an
$O(\log^3 n)$-round algorithm for uniformly sampling spanning trees from graphs with $O(n\log n)$ cover times.  These results are obtained by adding a load balancing component to the random walk algorithm of Bahmani, Chakrabarti and Xin (SIGMOD 2011) that uses the bottom-up ``doubling'' technique.
\end{abstract}

\section{Introduction}
Random spanning trees have been a fascinating area of mathematical study given their close connections to electrical circuits and random walks, dating back to Kirchoff in the 19th century.  Of particular interest is the Matrix-Tree theorem; usually credited to Kirchoff, although a more thorough account of its history is described in \cite{Kirchoff}.  This theorem states that the number of spanning trees of any graph can be found by taking the determinant of a minor of the graph Laplacian.

The connections between random walks and random spanning trees extend to efficient algorithms for randomly generating uniform spanning trees of a graph. 
This began with Aldous\cite{AAB} and Broder\cite{BAB} independently discovering that the set of edges used to first visit each vertex during a random walk (except the starting vertex of the walk which can be chosen arbitrarily) form a uniformly chosen spanning tree of the graph.  Since the expected time needed to visit every vertex of the graph, the \textit{cover time}, is known to be $O(mn)$ \cite{CoverTime}, for an $n$-vertex, $m$-edge graph, this immediately leads to an $O(mn)$ expected time algorithm for uniformly sampling spanning trees of a graph exactly.  Wilson found a faster random walk algorithm \cite{WilsonsAlg} for sampling spanning trees with an expected runtime of the average \textit{hitting time} of the graph; however, this algorithm still has an expected $\Theta(mn)$ runtime in the worst case.  Improvements have been made to the base Aldous-Broder algorithm. In particular, the problem with the algorithm is that while many distinct vertices are visited quickly at the beginning of the random walk, it can take a long time for the last few vertices to be visited.  This was addressed by the shortcutting method introduced by Kelner and Mądry \cite{MK}.  In its original form, this result only allowed for approximate sampling.  However, this has been extended to exact sampling by a trick from Propp, see \cite{madryOther,madryThesis}.  The shortcutting method was further improved by Mądry, Straszak, and Tarnawski \cite{madryOther} and by Schild \cite{shortcutAS}.  The result by Schild reduces the runtime to $O(m^{1+o(1)})$, for a graph with $m$ edges.  The high level idea of the shortcutting method is that once a part of the graph has been fully visited by a random walk, all future visits to that part of the graph are no longer relevant to the generated spanning tree.  Thus, instead of rewalking over parts of the graph that have already been visited, the random walk can take a \textit{shortcut}, jumping directly to a not yet fully visited part of the graph.  Finally, using Markov chain Monte Carlo (MCMC) methods rather than random walks, an almost-linear, $O(m \log^2 n)$ time, algorithm for approximately sampling spanning trees in the sequential setting has been shown by Anari, Liu, Gharan, Vinzant, and Vuong \cite{UpDownWalk}.

These impressive algorithmic advances for sampling random spanning trees are motivated by several applications of random spanning trees in theoretical computer science, including graph sparsification \cite{STandExpander, GraphSparsifier, DOLEV202321}, breakthroughs in approximation algorithms for the traveling salesman problem \cite{ATSP, STSP, NTSP}, and the $k$-edge connected multi-subgraph problem \cite{KarlinKleinGharanZhangSTOC2022}.

The focus of this paper is \textit{distributed} random spanning tree sampling. Specifically, we design our algorithms in the well-known \cc{} model. This model is a simple, bandwidth-restricted ``all-to-all'' communication model for distributed computing.
A wide variety of classical graph problems, including maximal independent set (MIS) \cite{ghaffari17_distr_mis_all_all_commun,GhaffariGKMRPODC18}, $(\Delta+1)$-coloring \cite{CzumajDPPODC2020}, ruling set \cite{HegemanPS14,CambusKPU23}, minimum spanning tree (MST) \cite{ghaffari16_mst_log_star_round_conges_clique,hegeman15_towar_optim_bound_conges_clique,LotkerPPPSPAA2003,jurdzinski18_mst_o_round_conges_clique,DBLP:conf/fsttcs/PemmarajuS16}, shortest paths \cite{doi:10.1137/19M1286955,DoryParterJACM2022}, minimum cut \cite{GhaffariNowickiPODC2018}, spanners \cite{ParterYogevDISC2018} and clique detection and listing \cite{DolevLP12,MatrixMult,censorhillel2024distributedsubgraphfindingprogress} have been studied in the \cc{} model.
The running times of the fastest algorithms for these problems range from $O(1)$ for MST to $O(n^{1-2/p})$ for detecting a clique of size $p$.
Algorithms in the \cc{} model are often the starting point for algorithms in more realistic ``all-to-all'' communication models for large-scale cluster computing, such as the $k$-machine model \cite{klauck15_distr_comput_large_graph_probl}, the Map Reduce model \cite{KarloffSuriVassilvitskiiSODA2010}, and the closely related Massively Parallel Computation (MPC) model \cite{AndoniNOYSTOC2014,BeameKSJACM2017,GoodrichSZISAAC2011}.

While there is vast literature in distributed computing on finding an arbitrary feasible solution to constraint satisfaction problems, e.g., finding a maximal independent set, a $(\Delta+1)$-coloring, or a spanning tree, there is relatively limited understanding of sampling from a distribution over the set of feasible solutions. 
For example, due to a series of papers over the last 2 decades \cite{LotkerPPPSPAA2003,hegeman15_towar_optim_bound_conges_clique,ghaffari16_mst_log_star_round_conges_clique,jurdzinski18_mst_o_round_conges_clique, nowickiMST}, the MST problem can be solved in just $O(1)$ deterministic rounds in the \cc{} model. But as far as we know, there is no work on sampling a random spanning tree in the \cc{} model.

There has been limited work on sampling combinatorial objects (e.g., colorings, independent sets, spanning trees) in other standard models of distributed computing such as \textsc{Congest} and \textsc{Local}; see \cite{FischerGhaffariDISC2018,CongestRandomWalk, WhatCanBeSampledLocally, distMetSampler, distJVV, LLLSampling, distSymBreaking, distFlipDynamics, exactDistSampling} for examples.  For instance, in \cite{WhatCanBeSampledLocally} the authors present a distributed version of the sequential Metropolis-Hastings algorithm for weighted local constraint satisfaction problems; subsequent work \cite{FischerGhaffariDISC2018,distSymBreaking} has improved this result.  However, while there has been some work on distributed sampling, there are still a lot of fundamental gaps in our understanding. This paper aims to fill some of these gaps.

\subsection{Preliminaries}
\label{section:Preliminaries}
Here we describe notation that will be used in subsequent sections.  The notation $\tilde{O}(f(n))$ refers to $O(f(n) \cdot \text{poly(log(n))})$.  We use $G$ to denote the input graph to our algorithm and following convention, $n$ and $m$ will denote the number of vertices and edges in $G$ respectively.  We use $\adj$ to denote the transition matrix of the random walk on $G$.  In particular, any vertex $a$ has equal probability,
i.e., $1/\text{degree}(a)$, of transitioning to any of its neighbors.  We use $\adj[a,b]$ to refer to the row-$a$, column-$b$ entry in $\adj$ and $\adj[a,*]$ to refer to row $a$ of $\adj$.  Finally, we take slight liberty with notation to let $(a,b)\in seq$ refer to the presence of consecutive elements $a,b$ in the sequence $seq$.

\subsection{Our Results}
\label{subsection:results}

\paragraph{Main Result: Sampling Random Spanning Trees in Sublinear Rounds}

Our main contribution in this paper is to show that the Aldous-Broder random spanning tree algorithm can be implemented in the \cc{} model in $o(n)$ rounds. We prove the following theorem\footnote{The requirement that the graph is unweighted can be slightly loosened.  We can allow edge weights to be positive integers bounded by $W = O(n^{\beta})$ for arbitrary constant $\beta$.  Here, the probability of a spanning tree is proportional to the product of its edge weights.  Likewise, the edge taken during each step of a random walk is chosen proportional to its edge weight.  The main reason edge weights need to be bounded is that the cover time of the graph is bounded by $O(W\cdot|V|\cdot|E|)$.  For simplicity, we will only focus on unweighted graphs.  With weighted edges, different choices will have to be made for some parameters in the algorithm.}.

\begin{restatable}{theorem}{mainresult}
\label{main-result}
There is an $\Tilde{O}(n^{1/2 + \alpha})$ round algorithm in the \cc{} model for approximately generating a uniform spanning tree of an arbitrary unweighted graph within total variation distance $\epsilon = \Omega(\frac{1}{n^{c}})$, for arbitrary $c>0$, from the true uniform distribution, where $O(n^\alpha)$ is the running time for matrix multiplication in the \cc{} (currently $\alpha = 0.157$).
In this algorithm, in every round, every machine performs polynomial-time local computations.
\end{restatable}

The most common technique for constructing random walks in ``all-to-all'' communication models such as \cc{} and the MPC model is the ``doubling technique'', a bottom-up technique which starts with many short walks and repeatedly stitches these together to build longer walks. 
At a high level, the idea is for every vertex $v$ to start an iteration possessing some number of walks of length $L$ originating at $v$. 
During the iteration, pairs of walks are ``matched'' and stitched together to create random walks of length $2L$. 

This technique is quite useful for constructing short, e.g., $O(\text{poly}(\log n))$ length, random walks, designed for PageRank estimation \cite{BCX,LackiMOSSTOC2020}, though it needs additional load balancing to be as efficient as possible. 
But the Aldous-Broder algorithm requires a random walk that covers the graph; such a walk has, in the worst case, expected $\Theta(mn)$ length for an $n$-vertex, $m$-edge graph. 
In order to construct a length-$\ell$ random walk, the algorithm of Bahmani, Chakrabarti, and Xin \cite{BCX} starts with every vertex $v$ holding $\ell$ length-1 random walks (i.e., random edges) originating at $v$. To implement the Aldous-Broder algorithm, $\ell$ needs to be $\Theta(mn)$, but for this setting of $\ell$, there is not enough bandwidth to even complete the first ``doubling'' iteration efficiently. 
The algorithm of \L{}\k{a}cki, Mitrovi\'{c}, Onak, and Sankowski \cite{LackiMOSSTOC2020} suffers from the same bottleneck; in fact, in this algorithm each vertex starts off holding even more length-1 walks initially.  Our algorithm is a significant departure from these well-known bottom-up techniques, featuring a novel \textit{top-down walk filling} approach.

At a very high level, we use the top-down approach to construct the walk phase by phase.
We impose a key restriction on the walk constructed in a phase: \textit{the walk in a phase contains $\sqrt{n}$ distinct vertices, not already visited in previous phases}.
Note that even with this restriction, the walk in a phase can be quite long, in fact $\Theta(m n)$ in length.
Visiting $\sqrt{n}$ distinct, new vertices in each phase implies that our algorithm uses $O(\sqrt{n})$ phases. The algorithm's $\tilde{O}(n^{1/2+\alpha})$ run-time arises from the fact that each phase takes roughly matrix multiplication time. Ensuring that the walk is truncated as soon as it visits the $\sqrt{n}$-th distinct vertex is non-trivial (see the discussion of ``Distributed Walk Truncation'' in Section \ref{section:overviewTechniques}), but this idea plays a crucial role in ensuring that there is enough bandwidth to efficiently perform a distributed sampling of fill in vertices that are needed to construct the walk in a phase. 
However, communicating these distributed samples to a leader machine for walk reconstruction turns out to be prohibitively expensive. We deal with this challenge by showing that the distributed samples can be compressed and communicated to the leader in such a way that the leader can produce a local sample with the same correct distribution. Somewhat remarkably, we reduce the problem of producing a local sample with the correct distribution to uniformly sampling a weighted perfect matching. This sampling can be performed by a local polynomial-time computation using the classical results of Jerrum, Sinclair, and Vigoda on approximating the permanent \cite{Permanent} and Jerrum, Valiant, and Vazirani \cite{JVV} on reducing approximate sampling to approximate counting in polynomial time.
After the first phase has been completed, it becomes critical for the efficiency of the algorithm, to avoid vertices already visited in previous phases. To implement this idea, we appeal to the Schur complement graph as used in \cite{Kyng_2017, shortcutAS} and shortcutting as used by Kelner and Mądry \cite{MK}.
While these ideas are familiar in the sequential setting in the context of randomly sampling spanning trees, our contribution is to show that these auxiliary graphs can be efficiently constructed using fast matrix multiplication in the \cc{} model.
We give a high level overview our algorithm, highlighting all of these techniques in more detail, in Section \ref{section:overviewTechniques}.

\paragraph{Remark.} We can adapt our algorithm to give exact rather than approximate samples.  This leads to a slower, though still $o(n)$ runtime of $\tilde{O}(n^{2/3+\alpha})$.  We sketch this approach in the \hyperref[Exact Sampling Appendix]{Appendix}.

\paragraph{Linear-Length Walks in Polylogarithmic Rounds} 

As mentioned earlier, there are algorithms \cite{BCX,LackiMOSSTOC2020} that use the ``doubling'' technique to construct random walks in a bottom-up manner, starting from a collection of length-1 walks. 
Unfortunately, these algorithms end up being inefficient because they are not inherently load-balanced. 
This is true even if our goal is to construct a $\Theta(n)$-length, random walk.  This is despite the fact that the \cc{} model has an overall bandwidth of $\Theta(n^2 \log n)$ bits, 
which is sufficient for each doubling iteration in \cite{BCX}, when $\ell = \Theta(n)$.  We present a load balanced version of the ``doubling'' technique and show how to efficiently construct relatively short random walks in the \cc{} model. The specific theorem we prove and its consequences are provided below.  
The most interesting instances of this theorem are walks of length $O(n\cdot \text{poly}(\log n))$, which we can construct in $O(\text{poly}(\log n))$ rounds, and walks of length $O(\text{poly}(\log n))$, which we can construct in $O(\log\log n)$ rounds. As described in \cite{BCX,LackiMOSSTOC2020}, walks of length $O(\text{poly}(\log n))$ are of particular interest for approximating PageRank.

\begin{restatable}{theorem}{secondResult}
\label{loadBalancing}
There is an algorithm for taking a random walk of length $O(\tau)$ in the \cc{} model that runs in 
\begin{itemize}
\item $O\left(\frac{\tau}{n} \log \tau \log n\right)$-rounds with high probability for $\tau = \Omega(n/\log n)$
\item $O(\log \tau)$-rounds with high probability for $\tau = O(n/\log n)$.
\end{itemize}
\end{restatable}

\begin{corollary}
For a graph with cover time $\tau$, we can sample a random spanning tree in $\Tilde{O}(\tau/n)$ rounds in the \cc{} model with high probability.
\end{corollary}
This result is interesting as many graphs have an $O(n \log n)$ cover time, notably expanders and Erd\"{o}s–R\'{e}nyi random graphs $G(n, p)$ with $p = \Omega(\log n/n)$ \cite{BAB, BK89, commuteTimes}.  It is also worth noting that for these graphs of near uniform degree, slightly better results can be obtained since the doubling algorithm is intrinsically load-balanced.  However, there exist dense, highly irregular graphs with $O(n\log n)$ cover times.  For example, $K_{n-\sqrt{n}, \sqrt{n}}$ can be seen to have $O(n \log n)$ cover time, applying the coupon collector argument.

\subsection{Overview of Our Algorithm and Key Techniques}
\label{section:overviewTechniques}
We first describe a very simple, sequential, recursive algorithm for generating a random walk; this will be a top-down algorithm and serve as our starting point.  By choosing the length $\ell$ of the walk to be $\Tilde{\Theta}(n^3)$, by the result of Aldous and Broder \cite{AAB, BAB}, we also obtain a random spanning tree with high probability.  We then present a first attempt at porting this algorithm to the \cc{} model.  Numerous challenges arise in porting our algorithm to the \cc{} model and we present a variety of new ideas to overcome these challenges. In our final attempt, we describe, at a high level, the first $o(n)$-round algorithm for sampling random spanning trees. For all of these algorithms we will choose $\ell$ to be a power of two for convenience.
\paragraph{A Sequential Algorithm}
\label{formulaPar}

Consider the following recursive algorithm that samples a uniformly distributed random walk of length $\ell$ starting at an input vertex, $s$. The algorithm first samples an end vertex $e$ of an $\ell$-length random walk that starts at $s$  and then feeds $s,e$ and $\ell$ to a recursive subroutine, that we call \textbf{Fill}. The subroutine \textbf{Fill} then samples a midpoint vertex $m$ of an $\ell$-length random walk that starts at $s$ and ends at $e$, using Bayes' rule.  More precisely, we want to sample the midpoint of an $\ell$ length random walk beginning at $s$ and conditioned on ending at $e$.  By the Markov property of random walks, the probability that vertex $v$ is the midpoint of such a walk is proportional to the probability that an $\frac{\ell}{2}$ length walk starting at $s$ ends at $v$ times the probability that another $\frac{\ell}{2}$ length walk starting at $v$ ends at $e$.  Then the subroutine recursively generates an $\ell/2$-length walk with $(s,m)$ as start-end vertices and another $\ell/2$-length walk with $(m,e)$ as start-end vertices. Finally, it joins the two $\ell/2$-length walks and returns the resulting $\ell$-length walk. Note that all of the vertex sampling during an entire run of the algorithm can be done using $O(\log \ell)$ distinct powers of the transition matrix, $\adj$, of the input graph.  This includes the endpoint $e$ and all of the midpoints.  We compute all of these matrices in advance, before entering the recursion. We give the complete outline of this algorithm in~\tcbref{box:seq1}.
 \begin{tcolorbox}[title=Outline 1. A sequential starting point,colback=white,colframe=black]
\refstepcounter{tcbbox}
\label{box:seq1}
\textbf{Top down filling algorithm }  \\
\textbf{Input:} Start vertex $s$ and walk length $\ell$\\
\textbf{Output:} Random walk of length $\ell$ beginning at $s$

\begin{enumerate}
    \item Compute $\adj,\adj^2,\adj^4,...,\adj^\ell$.
    \item Sample end vertex $e$ from $\adj^\ell[s,*]$
    \item Return \textbf{Fill}$(s,e,\ell)$.
\end{enumerate}

Subroutine: \textbf{Fill}\\
{Input:} Start vertex $s$, end vertex $e$, and walk length $\ell$\\
{Output:} Random walk of length $\ell$ beginning at $s$ and ending at $e$

  \begin{enumerate}
      \item If $\ell = 1$, Return $(s,e)$.
      \item Sample midpoint $m$ of an $\ell$-length walk starting at $s$ conditioned on ending at $e$.  Specifically, the unnormalized distribution is $\big(\adj^{\ell/2}[s,v]\cdot \adj^{\ell/2}[v,e]\big)_{v\in V}$.
      \item Return $\textbf{Join}\big(\textbf{Fill}(s,m,\frac{\ell}{2}), \textbf{Fill}(m,e,\frac{\ell}{2})\big)$.
  \end{enumerate}
\end{tcolorbox}
 
\vspace{1em}

\noindent
\textbf{First attempt at a \cc{} algorithm.}
We now focus on transforming the above outline of a sequential random walk sampling algorithm into a distributed algorithm in the \cc\ model. We begin by making a natural, structural shift -- in each level of the recursion all of the work will be done in parallel across multiple machines. Recall that during each level of the recursive sequential algorithm a set of midpoints is generated and inserted into the walk.  Our goal is to delegate this midpoint generation process to multiple machines and then to centrally collect the generated midpoints and insert them into their respective correct locations.  In our notation, we let $W$ be the walk under construction.  Note that $W$ is a \textit{partial walk} with ``holes'' that need to be filled.
We let $W[j]$ refer to the vertex at index $j$ of $W$; this indexing takes into account the yet-to-be-filled holes.  We let $W_i$ refer to the partial walk that has been constructed before level $i$ of the algorithm.  When we refer to adjacent vertices in $W_i$, we mean vertices that are only separated by holes.

We assume that we are given a start vertex $s$ and designate the machine $\sfM$ holding $s$ as the \textit{leader} machine. We initialize the walk with its first vertex, i.e., set $W[0] = s$. During the course of the algorithm, $\sfM$ is the machine that will store the random walk as it is constructed. 
We then compute the relevant powers of the transition matrix $\adj$ and sample the end point $e$ using the row $s$ of the matrix $\adj^{\ell}$.
We insert $e$ into $W$, as its last vertex, i.e., we set $W[\ell] = e$. We are now ready to fill in the rest of the walk level-by-level by sampling midpoints. 
Now assume that the algorithm is just about to enter the $i$-th level. At this time, $\sfM$ holds $W_i$, a partial walk with $2^{i-1}+1$ evenly spaced vertices.  Note that the ``gap'' between consecutive filled in vertices in $W_i$ is $\frac{\ell}{2^{i-1}}$.  
In other words, this is the number of vertices we need to fill between consecutive vertices in $W_i$.
Our goal is now to sample a midpoint between each consecutive pair of vertices in $W_i$.  For each distinct consecutive pair, $(p,q)$ in $W_i$, we assign a machine $\sfM_{p,q}$ to sample the midpoints for all occurrences of this pair in $W_i$.  This algorithm structure is motivated by the fact that once we fix a walk length, the probability distribution of a midpoint is completely determined by the endpoints $(p, q)$. Thus it makes sense for a single machine $M_{p,q}$ to sample all the midpoints for that pair.   
The outline of this first attempt at a \cc{} algorithm is presented in \tcbref{boxdist}.
Note that this description skips some lower level details, e.g., the computation of the powers of the $\adj$ in the \cc{} model. All of these details are provided in Section \ref{genAlg}.
The reader is also encouraged to consult Figure \ref{fig:midPointPlacement} to better understand this outline.

As presented, this algorithm faces two major bottlenecks in the \cc{} model. 
\begin{description}
    \item[Bottleneck 1:]First, in step $5(b)$, there could be many, i.e., $\omega(n)$, distinct start-end pairs $(p,q)$. This means that we cannot find a \textit{distinct} machine for each endpoint pair. We could address this problem by making a machine responsible for a collection of start-end pairs. This idea faces a fundamental bandwidth bottleneck. To correctly sample a midpoint to place within the start-end pair $(p, q)$, machine $\sfM_{p,q}$ needs to receive the probability distribution over all vertices in the graph occurring as midpoints for this start-end pair. This constitutes $n$ words of information. Thus, for a machine to receive all necessary probability distributions quickly, we can only assign a few distinct start-end pairs to a machine.
     
    \item[Bottleneck 2:]Secondly, the bandwidth doesn't exist for each machine to efficiently send $\Pi_{p,q}$ back to $\sfM$, as $c_{p,q}$ may be as large as $\Tilde{\Theta}(n^3)$.
\end{description}

\begin{tcolorbox}[title=Outline 2.  
A first \cc{} algorithm,colback=white,colframe=black]
\refstepcounter{tcbbox}
\label{boxdist}
  \textbf{A first attempt at a \cc{} algorithm}\\
  \textbf{Input:} Start vertex $s$ and walk length $\ell$\\
  \textbf{Output:} Random walk of length $\ell$ beginning at $s$

  \begin{enumerate}
       \item Set $\sfM$, the machine holding $s$, as the leader machine; set $W[0]=s$.
       \item Compute $\adj^2,\adj^4, ..., \adj^\ell$.
       \item $\sfM$ samples end vertex $e$ from $\adj^\ell[s,*]$ and sets $W[\ell] = e$.
       \item $W_1 = W$
       \item For level $i = 1$ to $\log \ell$
        \begin{enumerate}
            \item $c_{p,q}:=$ No.~of times the vertices $p,q$ occur adjacently in $W_i$.                    
            \item For each distinct $(p,q)\in W_i$, $\sfM$ designates a machine $\sfM_{p,q}$.

            \item $\sfM$ sends $c_{p,q}$ to $\sfM_{p,q}$.  ~~~~~~~~~~~~~~~~~~~~~~~~~~~~~~~~~~~~~~~~~~~~~~~~~~~~
            \item $\sfM_{p,q}$ sends $\Pi_{p,q}$ to $\sfM$.
            \item $\sfM$ places the midpoints from each sequence in the corresponding positions in $W_i$ to obtain $W_{i+1}$. 
        \end{enumerate}
        \item Return $W$.
  \end{enumerate}
   
\end{tcolorbox}

\noindent
\textbf{A second attempt at a \cc{} algorithm.} Another shift in perspective helps us overcome these bottlenecks.  We leverage the fact that we only need vertices' first-visit edges, and can therefore skip over vertices previously visited during the random walk. 
On the basis of this idea, we split the generation of the random walk into phases.  In each phase, we generate a new segment of the walk containing $\sqrt{n}$ \textit{distinct} vertices, starting from the final vertex $s$ visited in the previous phase.  
Except for $s$, all vertices visited in the current phase are new, i.e., have not been visited in previous phases.
By limiting the number of distinct vertices visited in a phase, we resolve the first bottleneck mentioned above. 
Specifically, $\sqrt{n}$ distinct vertices in a walk segment implies $n$ distinct start-end $(p, q)$ pairs, which in turn implies the need for at most $n$ machines $\sfM_{p,q}$ for midpoint generation.  To avoid repetitively visiting the same vertices in multiple phases, we use short-cutting in later phases (described in more detail further below).  This ensures that the number of phases is $\Theta(\sqrt{n})$.  Each phase requires roughly matrix multiplication time, leading to the claimed runtime.  In order to visit $\sqrt{n}$ distinct vertices per phase, we use distributed walk truncation (described in more detail below).

\paragraph{Distributed Walk Truncation.}
In each phase, our algorithm ensures that the leader machine $\sfM$ only receives the midpoints up to some truncation point of its partial walk. The truncation point is such that the number of distinct vertices within the prefix of the partial walk up to the truncation point, including the new midpoints, does not exceed $\sqrt{n}$. This truncation idea seems daunting to implement at first because $\sfM$ has to determine a truncation point for the walk without even looking at the sequences $\Pi_{p,q}$ of new midpoints generated by machines $\sfM_{p,q}$. 
Recall the second obstacle mentioned above: there is not enough bandwidth for $\sfM$ to receive the full midpoint sequences $\Pi_{p,q}$ from the machines $\sfM_{p,q}$.
We solve this issue by using distributed binary search (see Section \ref{subsubsec: truncated} for more details). 
Specifically, $\sfM$ guesses a truncation point and distributes requests for midpoint generation
for start-end pairs in the partial walk up to the truncation point. 
We show how machines receiving midpoint generation requests can aggregate their answers efficiently to
determine if midpoint generation up to the current truncation point will cause the $\sqrt{n}$ budget on the number of distinct vertices in the partial walk prefix to be exceeded. 
Accordingly, $\sfM$ can adjust its guess of truncation point and this binary search process continues
until $\sfM$ finds the largest truncation point such that the new partial walk up to that truncation point contains $\sqrt{n}$ distinct vertices.  Note that by choosing $\ell = \Tilde{\Theta}(n^3)$, we can ensure that such a truncation point will exist with high probability, since the entire walk will visit all $n$ vertices.

\paragraph{Sampling Random Perfect Matchings.}
Now, we focus on the second challenge.  Recall that the identity of the newly generated midpoints and \textit{their ordering along the walk} is held in a distributed fashion by the machines $\sfM_{p,q}$.  Unfortunately, this combined information (contained in all the $\Pi_{p,q}$'s) is prohibitively large to be communicated back to the leader machine $\sfM$.  We show that it is possible to compress this data and send it to $\sfM$, so that $\sfM$ is able to re-sample the walk \textit{locally}.  Note that this local sampling will not recover the exact walk defined by the midpoint sequences $\Pi_{p,q}$. However, it will generate a walk with the correct probability that uses the same midpoints as the one defined by the $\Pi_{p,q}$'s.  
We use two key ideas. First, $\sfM$ only collects the multiset, $\mathbb{M}$, of the midpoints, without any information about their position in the walk.  Second, using $\mathbb{M}$, $\sfM$ can locally sample a weighted perfect matching to resample the locations of the midpoints.
It is clear that the first idea allows us 
to save on communication bandwidth, especially because there are $\sqrt{n}$ distinct vertices in each walk segment constructed in a phase.  Using perfect matchings to resample a walk is a new idea and we further expand upon how $\sfM$ executes the placement process of collected midpoints. 

Machine $\sfM$ locally constructs a complete bipartite graph $\mathcal{B}$, where the vertex set on one side of the bipartition is $\mathbb{M}$ (the sampled midpoints) and the vertex set on the other side, which we denote $\mathbb{P}$, is the set of midpoint positions in the walk\footnote{The precise complete bipartite graph that we construct is slightly different and is defined in Section \ref{subsubsec: truncated}. Specifically, for technical reasons, care has to be taken with the chronologically final midpoint in the walk  (see Lemma \ref{couplingLemma}). So the final midpoint and the final midpoint position are excluded from the eventual complete bipartite graph we use. We gloss over these details until Section \ref{subsubsec: truncated}.}.
Note that a perfect matching $\mu$ of $\mathcal{B}$ corresponds to a valid assignment of midpoints in $\mathbb{M}$ to each of the midpoint positions in the partial walk.
In the example in Figure \ref{fig:midPointPlacement}, since the midpoint multiset $\mathbb{M} = \{1, 2, 3, 3, 3, 3, 4, 4\}$ and these midpoints are to be placed in 8 positions, the (unweighted) graph $\mathcal{B} = K_{8, 8}$.
To each edge $(x, y)$ in $\mathcal{B}$, machine \sfM{} assigns a weight proportional to the probability of the sampled midpoint $x \in \mathbb{M}$ appearing as a midpoint for the start-end pair corresponding to midpoint position $y$.  Note that different midpoint positions could correspond to the same start-end pair.  There are at most $O(\sqrt{n})$ distinct vertices that are either already placed in the walk or sampled as new midpoints.  Since the edge weights of $\mathcal{B}$ depend only on transition probabilities between these vertices, to construct $\mathcal{B}$, $\sfM$ only needs to collect a $\sqrt{n}\times\sqrt{n}$ submatrix of a power of the transition matrix.  This can be done in $O(1)$ rounds.

Once the edge-weighted, complete bipartite graph $\mathcal{B}$ is defined, we define the weight of a perfect matching $\mu$ of $\mathcal{B}$ as the product of the weights of all of the edges in $\mu$. 
The motivation for all of these definitions is the connection we show in Lemma \ref{matchingLemma}: \textit{the weight of a perfect matching $\mu$ is proportional to the probability of independently sampling the assignment of midpoints in $\mathbb{M}$ to each of the midpoint positions in $\mathbb{P}$, specified by $\mu$.}
So to sample such an assignment of midpoints to midpoint positions, we just need an algorithm 
for sampling a random perfect matching $\mu$ of $\mathcal{B}$ with probability proportional to the weight of $\mu$.  We show that the leader machine $\sfM$ can do this sampling approximately in polynomial time by appealing to the classical algorithm of Jerrum, Sinclair, and Vigoda on approximating the permanent of a matrix\cite{Permanent}.  Since the permanent can be used to compute the weighted sum of perfect matchings of a weighted bipartite graph, we can then use the algorithm of Jerrum, Valiant, and Vazirani \cite{JVV} to reduce approximate sampling to approximate counting in polynomial time.

\begin{figure}[h]
    \centering
    \includegraphics[width=.95\textwidth]{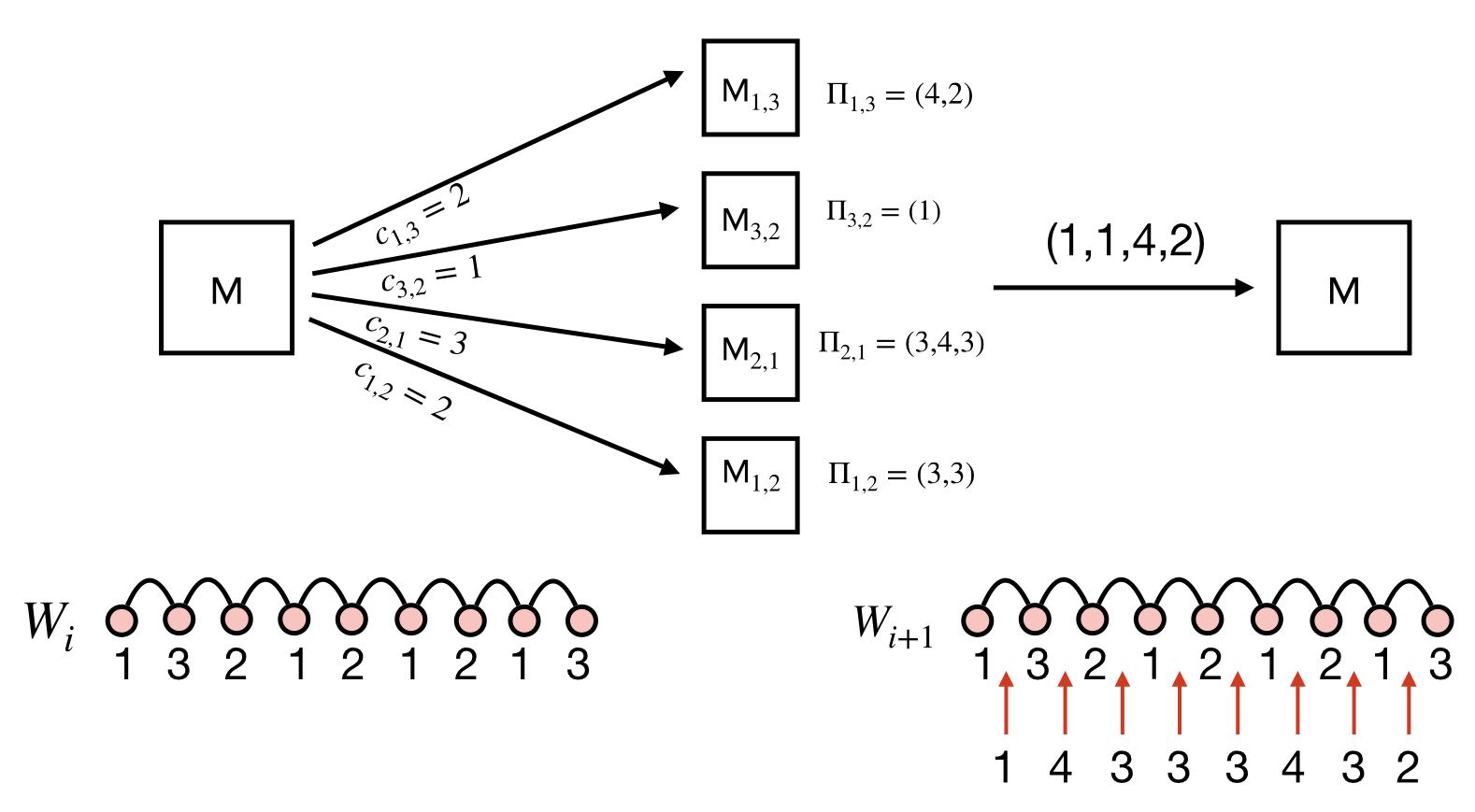}
    \caption{This figure illustrates how machine \sfM{} adds midpoints to the walk $W_i$ to obtain $W_{i+1}$.  Note that in $W_i$ there exist the distinct start-end pairs: $(1,3),(3,2),(2,1),(1,2)$.  \sfM{} sends the count $c_{p,q}$ of each start-end pair $(p, q)$ to the machine $\sfM_{p,q}$ responsible for that pair.  The machine $\sfM_{p,q}$ then generates a sequence $\Pi_{p,q}$ containing $c_{p,q}$ midpoints. For example, machine $\sfM_{1, 3}$ generates $\Pi_{1,3} = (4, 2)$ indicating that midpoint 4 is to be inserted within the first $(1, 3)$ and midpoint 2 is to be inserted within the second $(1, 3)$. Collectively, the machines $\sfM_{p,q}$ send the multiset $\mathbb{M}= \{1, 2, 3, 3, 3, 3, 4, 4\}$ (written in short, as the vector $(1, 1, 4, 2)$) of generated midpoints to \sfM.  Finally, \sfM{} samples a perfect matching between the sampled midpoints and midpoint positions in the walk, shown with red arrows.  The midpoints are then placed in the walk in these selected indices.  Here we ignore the subtlety of ensuring that $W_{i+1}$ contains at most $O(\sqrt{n})$ distinct vertices and also the care that needs to be taken around the final midpoint.}
    \label{fig:midPointPlacement}
\end{figure}

\paragraph{Using Schur Complement and Shortcut Graphs.}
The final piece of our algorithm is the short-cutting process. Recall that the idea is to skip over vertices in each phase that were visited in prior phases.  It is evident that short-cutting applies only \textit{after} the first phase. 
We use two types of graphs for this -- the \emph{Schur 
complement} graph and \emph{shortcut graph}. 

Consider the situation just before a phase and let $S$ be the set of vertices in the input graph $G$ that have not yet been visited along with the last vertex visited in the previous phase. The \emph{Schur complement} of $G$, denoted by $\schur(G,S)$, is an edge-weighted  undirected graph, with vertex set $S$ such that, roughly speaking, taking a random walk on $\schur(G,S)$ yields the same distribution as
taking a random walk on $G$ and then removing the vertices in $V \setminus S$.  For all phases except the initial phase, we replace $\adj$ with $\scc$, the transition matrix for a random walk on $\schur(G,S)$.  This allows us to take a random walk on $G$ while skipping over vertices in $V \setminus S$.

Problematically, the walk on $\schur(G, S)$ returns edges from $\schur(G, S)$. For our spanning tree sampling, we need the first visit edges in the original graph $G$. For this, using ideas from~\cite{MK}, we need another derivative graph on $G$, that we call the \emph{shortcut graph}. Interestingly, we show that using fast distributed matrix multiplication algorithms, both of these graphs can be computed efficiently in each phase.  

Note that while the cover time of an arbitrary weighted graph can be large, the cover time of $\schur(G,S)$ is always bounded by the cover time of $G$. Thus even beyond the first phase we can set $\ell = \Tilde{\Theta}(n^3)$ and still visit at least $\Theta(\sqrt{n})$ distinct vertices in $\schur(G,S)$ with high probability.
The outline below describes an arbitrary phase of our final algorithm.

\newpage $ $
\begin{tcolorbox}[title=Outline 3.  
Our final algorithm,colback=white,colframe=black]
\label{box:dist2}

  \textbf{One phase of final \cc{} algorithm}
  \begin{enumerate}
        \item $W[0] = s:=$ final vertex visited in the previous phase 
        \item $S:=$ Set of unvisited vertices $\cup~\{s\}$;~~ $\scc:=$ transition matrix of $\schur(G,S)$
        \item $\ell := \Tilde{\Theta}(n^3)$
        \item Compute $\scc^2,\scc^4, ..., \scc^\ell$ and sample $W[\ell]$ from $\scc^\ell[s,*]$.

        \item For level $i = 1$ to $\log \ell$
        \begin{enumerate}
            \item  \textbf{Leader }$\sfM$ \textbf{requests midpoints} (cf. Outline 2 step 5(a)-(c)) from designated machines. 
            \item Machines execute \textbf{midpoint sampling} (cf. Outline 2 step 5(d)) in response to $\sfM$'s request.
            \item $\sfM$ uses binary search to find the correct truncation point and truncates $W$
            \item $\sfM$ receives the multiset $\mathbb{M}$ of used midpoints
            \item $\sfM$ samples a random weighted perfect matching between $\mathbb{M}$ and midpoint positions
            \item $\sfM$ places midpoints according to the sampled matching
        \end{enumerate}
        \item The first visit edge to each newly visited vertex is sampled using \shortcut$(G,S)$.
  \end{enumerate}
\end{tcolorbox}

\subsection{Barriers to Further Improvements}
Non-trivial lower bounds in the \cc{} model are extremely difficult to obtain because of connections to long-standing open problems in circuit complexity \cite{DruckerKuhnOshmanPODC2014}.
As a result, it is difficult to get insight into how close to optimal our algorithm might be.
However, there seem to be significant challenges to improving either the more traditional bottom-up doubling approach or our top-down walk filling approach for sampling random spanning trees in the \cc{} model.  We discuss several potential directions for improvement and the corresponding challenges below.
It is also worth noting that many candidate algorithms do not sample spanning trees from the uniform distribution. For example, given that an MST can be constructed in $O(1)$ rounds in the \cc{} model, one might be tempted to randomly assign weights, say in $[0, 1]$, to edges and construct the MST on this edge-weighted graph. However, the distribution of the spanning tree constructed in this manner is well known to differ from the uniform distribution \cite{Goldschmidt_2018}.

\begin{description}
\item[Direction 1:] \textbf{Reduce the total running time of our top-down approach from $\tilde{O}(n^{1/2+\alpha})$ to $\tilde{O}(n^{1/2})$ rounds.}  
Recall that the $\tilde{O}(n^\alpha)$ factor in the running time arises from having to do matrix multiplication in each of the $O(n^{1/2})$ phases of the algorithm. 
We use matrix multiplication to compute the shortcut and Schur complement graphs at the end of each phase.
Getting rid of the $\tilde{O}(n^\alpha)$ factor (while keeping the rest of our approach the same) amounts to computing the shortcut and Schur complement graphs in less than matrix multiplication time.
This seems like an interesting problem in its own right and it might be approachable using
distributed Laplacian solvers \cite{CCLaplacianSolver}.  
In fact, Kelner and Madry \cite{MK} demonstrate that we can use Laplacian solvers to quickly compute shortcut probabilities.
However, there still seem to be two problems.  First, their result is only efficient when they limit themselves to computing shortcut probabilities between clusters of a graph where each cluster has only a few outgoing edges.  In our algorithm there is one large cluster with no bound on the outgoing number of edges and the remaining clusters contain single vertices.  Without a bound on inter-cluster edges, a solver would have to be able to solve many systems in parallel.  Second, even if this problem could be solved, we also need to compute several powers of the transition matrices of both the Schur complement and shortcut graphs.  Let $\adj$ be the transition matrix for a random walk on $G$ and $\scc$ be the same for $\schur(G,S)$.  In particular, there does not appear to be an easy way (i.e., taking less than matrix multiplication time) to compute $\scc^k$ from $\adj^k$.  
If such an efficient algorithm were available, we could have computed powers of $\adj$ at the beginning of the algorithm and then computed corresponding powers of $\scc$ more efficiently at the start of each subsequent phase.
Note that $\schur(G,S)^k \neq \schur(G^k, S)$, where here the power of a graph refers to the weighted graph given by the $k$-th power of its adjacency matrix.

\item[Direction 2:] \textbf{Reduce the total running time of our top-down approach from $\tilde{O}(n^{1/2+\alpha})$ to $\tilde{O}(n^\alpha)$ rounds.}  
In the sequential setting, spanning trees can be generated in matrix multiplication time \cite{MMST1, MMST2} (though these results assume matrix multiplication time is $\Omega(n^2)$).
We might be able to achieve this running time in the \cc{} model, if we weren't limited to only visiting $O(\sqrt{n})$ distinct vertices in each phase. 
Keeping other aspects of our algorithm the same, visiting more than $\omega(\sqrt{n})$ distinct vertices in a phase, implies the need to generate walk midpoints for $\omega(n)$ distinct start-end pairs.
The \cc{} model seems to fundamentally lack the bandwidth to efficiently generate such a large number of midpoints.

\item[Direction 3:] \textbf{Implement the bottom-up doubling approach efficiently.}
We motivate our top-down approach by identifying fundamental bottlenecks that the bottom-up doubling 
(as in \cite{BCX,LackiMOSSTOC2020}) runs into.
We elaborate on this point a bit more here.
Clearly, it takes only $O(\log n)$ iterations of the doubling algorithm to build an $\tilde{O}(n^3)$ length walk. The challenge is implementing each doubling iteration efficiently.  For the version of the doubling approach considered in \cite{BCX}, there seems to be a communication barrier even in the very first doubling iteration. Initially each machine holds $\Theta(n^3)$ length-1 walks and needs to receive an additional $\Theta(n^3)$ length-1 walks to merge with the walks it holds. This is clearly much higher than the $O(n)$ words bandwidth per machine.  Furthermore, to maintain correctness, each walk also needs to be tagged with an index, so we can't compress the information by having each machine receive a multiset of walks (i.e., a set of walks with counts), rather than a list.  Even if we didn't need to maintain indices, as in the algorithm of \cite{LackiMOSSTOC2020} which uses more total length-1 walks than \cite{BCX}, we still seem to face a problem as the number of distinct $k$-length walks starting from a vertex grows exponentially in $k$. Therefore, even communicating a multiset would require much more than $O(n)$ words after a few iterations.

\item[Direction 4:] \textbf{Designing a conceptually simpler $o(n)$-round algorithm.}
Theorem \ref{loadBalancing} tells us that a length-$n$ random walk can be computed in $O(\log^2 n)$ rounds with
high probability. This raises the natural question of how many distinct vertices a length-$n$ random walk visits. It turns out that Barnes and Feige \cite{FB} answer precisely this question and show that a length-$n$ walk visits $\Omega(n^{1/3})$ vertices, raising the possibility that the graph might be covered in $O(n^{2/3})$ phases. Unfortunately, the Barnes-Feige bound is known to hold only for unweighted graphs. After the first phase, we need to work with the edge-weighted Schur complement graph, and for such graphs, no bound similar to the Barnes-Feige bound is known.
It is worth emphasizing that even if this direction were successful, it would lead to an algorithm
whose running time would be worse than $\tilde{O}(n^{1/2+\alpha})$. 
However, this approach is conceptually simpler than our current approach and might be more easily amenable to improvements.
\end{description}

\subsection{Other Related Work}
\label{subsected:relatedWork}
In addition to the related work already discussed, a few other results are noteworthy in the parallel and distributed settings.  Teng \cite{PRAMWalk} showed that in the EREW PRAM model we can take a length $\Theta(n^3)$ walk, and hence sample a spanning tree, in $O(\log n)$ rounds.  Anari, Hu, Saberi, and Schild showed that for weighted spanning trees the problem is still in RNC \cite{ParWeightedSpanningTree}.  In the much weaker \congest{} model, Sarma, Nanongkai, Pandurangan, and Tetali \cite{CongestRandomWalk} show that sampling a spanning tree is possible in $\Tilde{O}(\sqrt{m} \cdot D)$ rounds, where $D$ is the diameter of the graph.  In this model, machines can only send limited bandwidth messages to their neighbors, rather than to every machine in the network.  Finally, there is a distributed algorithm for estimating the number of spanning trees by Lyons and Gharan \cite{STCounting}.  Given the well-known reduction from sampling problems to counting problems \cite{JVV}, this may be useful for sampling spanning trees.  However, in the case of spanning trees, it is not clear how to parallelize this reduction.

In addition to the doubling results mentioned earlier, one other doubling result of interest is by Luo \cite{RadarPush} and extended into a journal version by Luo, Wu, and Kao \cite{LuoWuKaoInfoSciences2022}.  However, the result is focused only on sampling the endpoints of random walks, for PageRank, instead of generating the entire walk.  Also, the algorithm requires a higher communication bandwidth than our version of the \cc{} model allows and doesn't seem to work with long walks.

Another interesting approach to randomly sampling spanning trees is to first randomly sparsify the input graph and then draw a random spanning tree from the sparsified graph.  Such an approach in the sequential setting is described by Durfee, Peebles, Peng, and Rao \cite{SparsifierSampler}.  However, their approach only allows random spanning tree sampling with constant (or slightly sub-constant) total variation distance, instead of inverse polynomial total variation distance from the uniform distribution.  Furthermore, independent of the issue with total variation distance, it is not clear whether their sparsification algorithm can be implemented efficiently in the \cc{} model.

\subsection{Congested Clique Model}
\label{model}
Given an $n$-vertex input graph $G = (V, E)$, the underlying communication network is the size-$n$ clique $K_n$, along with a bijection between vertices in $G$ and the machines in the communication network. Thus each machine hosts a distinct vertex and all edges in $G$ incident to that vertex.  

Machines have unique IDs of length $O(\log n)$ bits and without loss of generality for our algorithms, we say the machines have IDs $1$ through $n$ and the machine with ID $i$ holds vertex $i$.
Computation in the model proceeds in synchronous rounds.  Each round begins with each machine performing unbounded local computation.  A round then ends with each machine sending a (possibly different) message of length $O(\log n)$ bits to every other machine.  Note that a single message is able to encode a constant number of vertices or edges.
The time complexity of an algorithm is then measured by the number of rounds used.  While machines are theoretically allotted unlimited time for local computation, it is preferred that this time is polynomial in terms of the input graph size.  Our algorithms only require polynomial time local computation.

Lenzen \cite{Lenzen13} proved that in $O(1)$ deterministic rounds it is possible for every vertex to send and receive a total of $O(n)$ messages, regardless of the specific destination of each message.  Thus we take the more general view in this model that at each round a machine is able to send and receive a total of $O(n)$ messages instead of placing limits on bandwidth between pairs of machines.

Of particular interest to us is matrix multiplication in the \cc{}.  When $n \times n$ square matrices are distributed among machines so that machine $i$ holds row $i$ in each matrix, Censor-Hillel, Kaski, Korhonen, Lenzen, Paz, and Suomela showed that matrix multiplication can be carried out in $O(n^{\alpha})$ rounds \cite{MatrixMult} (currently $\alpha = .157$) \cite{NewMatrixMult}.

\subsection{Schur Complement and Shortcut Graphs}
We now introduce two derivative graphs of the original graph $G$, called the \textit{Schur complement} graph and the \textit{Shortcut} graph.
As mentioned earlier, the Schur complement graph will be used to skip over vertices already visited by the random walk constructed in prior phases. 
The shortcut graph, which is closely related to the Schur complement graph, is used in the setting where we have taken a walk on the Schur complement graph, but we wish to recover the first visit edge of a vertex in the underlying walk in $G$.

Schur complement is a notion from linear algebra that arises during Gaussian elimination. 
We consider an $n \times n$ matrix $M$ and a subset $S \subset [n]$ of the index set. Let $\overline{S}$ denote $[n] \setminus S$ and, without loss of generality, take $\overline{S}$ to be the first $|\overline{S}|$ indices in $[n]$. Borrowing notation from Section 2.3.3 in \cite{Kyng_2017}, we can rewrite $M$ as
$$M=\begin{bmatrix}M_{\overline{S},\overline{S}}&M_{\overline{S},S}\\M_{S,\overline{S}}&M_{S,S}\end{bmatrix}.$$
When $M_{\overline{S},\overline{S}}$ is invertible, we define the \textit{Schur complement of $M$ onto $S$} as
$$\schur(M, S) = M_{SS} - M_{S,\overline{S}} \cdot (M_{\overline{S},\overline{S}})^{-1} \cdot M_{\overline{S},S}.$$

This definition can be ported to graphs by noting that the Schur complement operation is closed for the class of graph Laplacians. Recall that the \textit{Laplacian} $L(G)$ of a simple, undirected, weighted graph $G = (V, E, w)$ with $V = [n]$ is a $n \times n$ matrix with entries 
$$
L(G)[i, j] =
\begin{cases}
\sum_{e: i \in e} w(e), & \text{if } i = j\\
-w(e), & \text{if } i \not= j \text{ and } e = \{i, j\} \in E\\
0, & \text{otherwise}
\end{cases}
$$
According to Fact 2.3.6 in \cite{Kyng_2017}, for any simple, undirected, weighted graph $G = (V, E, w)$ and vertex subset $S \subset V$, $\schur(L(G), S)$ is itself a Laplacian of a simple, undirected, weighted graph defined on the subset of vertices $S$. 

\begin{definition}[Schur Complement] 
For any simple, undirected, weighted graph $G = (V, E, w)$ and vertex subset $S \subset V$, the \textit{Schur complement graph of $G$ onto $S$}, denoted $\schur(G, S)$, is the simple, undirected, weighted graph $H$ with vertex set $S$ such that $L(H) = \schur(L(G), S)$.
\end{definition}
The motivation for defining $\schur(G, S)$ is that
taking a random walk in $G$ starting at $v \in S$ and only looking at the subsequence of visits to vertices in $S$ is identical to taking a random walk in \schur$(G,S)$ starting from the vertex $v \in S$.  For a weighted graph, a random walk refers to choosing the edge leaving a vertex proportional to that edges weight.  More precisely, as shown in Theorem 2.4 in \cite{shortcutAS}, the distributions over the sequences of vertices in $S$ generated by the two random walks are identical.  It is also possible to define the graph $\schur(G, S)$ implicitly (and more intuitively) by specifying the transition matrix, $\scc$, of a random walk on the graph.

\begin{definition}[Schur Complement Transition Matrix]
    For any $u, v \in S$, $\scc[u,v]$ gives the probability that $v$ is the first vertex in $S \setminus \{u\}$ that a random walk started at $u$ in $G$ visits.
\end{definition}

Figure \ref{derivativeGraphFigure} provides an illustration of the definition of $\schur(G, S)$.

We now define a second derivative graph, closely related to the Schur complement graph, that we call the \textit{shortcut} graph.
The shortcut graph is used in the setting where we have taken a walk on $\schur(G, S)$, but we wish to recover the first visit edge of a vertex in the underlying walk on $G$.  It is convenient to define the shortcut graph implicitly by specifying its transition matrix.

\begin{definition}[Shortcut Graph] The shortcut graph, \shortcut$(G, S)$, is a weighted directed graph with vertex set $V$ and transition matrix $\sct$ (for the random walk on \shortcut$(G,S)$) defined as follows.  Consider a random walk on $G$ starting at a vertex $u \in V$: $x_0 = u, x_1, x_2, ...$.  Let $j = \min \{i > 0 | x_i \in S\}$ be the index of the first vertex in the walk, after $x_0$, belonging to $S$.  Then $\sct[u,v] = \Pr[x_{j-1} = v]$ is the probability that $v$ appears just before $x_j$ in this walk. 
\end{definition}

Suppose that we have a single random transition in $\schur(G, S)$ from vertex $u \in S$ to vertex $w \in S$. Note that this transition corresponds to a random walk from $u$ to $w$ in $G$.  Suppose we wish to recover the edge $(v, w)$ used to visit $w$ in this underlying walk on $G$.
The shortcut graph provides a probability distribution over neighbors of $w$, that can be used with Bayes' rule to sample
$v$. Specifically, $\sct[u, v] \cdot \frac{1}{\deg_S(v)}$ provides an (unnormalized) distribution over neighbors of $w$ for correctly sampling the edge $(v, w)$. This will be explained in detail in Section \ref{futurePhases}.

\begin{figure}[h]
    \centering
    \includegraphics[scale = .75]{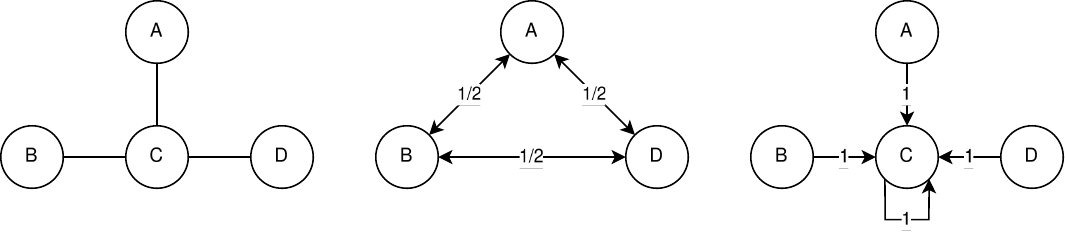}
    \caption{This figure illustrates both derivative graphs of $G$.  On the left is the original graph $G$.  In this example $S = \{A,B,D\}$.  In the center is \schur$(G,S)$.  Finally, on the right, is \shortcut$(G,S)$.  The labels on the edges give the transition probabilities. Note that the Schur complement graph contains uniform transitions between every vertex.  This is because a random walk started at $A$ (for instance) is equally likely to visit $B$ before $D$ or vice versa.  In the shortcut graph every vertex always transitions to $C$ since $C$ is always visited directly before a visit to a vertex in $S$ (except possibly at time 0).}
    \label{derivativeGraphFigure}
\end{figure}

\subsection{Sampling Weighted Perfect Matchings}
We now describe how sampling a weighted perfect matching of a bipartite graph can be achieved in polynomial time.  Recall we define the weight of a perfect matching as the product of the weight of its edges.  The sum of the weight of every perfect matching is given by the permanent of the biadjacency matrix of the bipartite graph.  Jerrum, Sinclair, and Vigoda gave an FPRAS for approximating the permanent \cite{Permanent}.  Further, using the sampling to counting reduction of Jerrum, Valiant, and Vazirani \cite{JVV}, we can then approximately sample perfect matchings proportional to their weight.  In particular, sampling within total variation $\delta$ can be done with a runtime that is polynomial in both the size of the bipartite graph and $\ln(\frac{1}{\delta})$.

\section{Sampling Spanning Trees in Sublinear Rounds}\label{genAlg}
Our algorithm will proceed in phases.  In each phase (except possibly the final phase), with high probability, an additional $\Theta(\sqrt{n})$ distinct vertices will compute their first visit edges in the random walk.  
This leads to a total of $O(\sqrt{n})$ phases being required to generate a uniform spanning tree.  
We spend (roughly) matrix multiplication time per phase, i.e., $\tilde{O}(n^{\alpha})$ rounds, for a total running time of $\Tilde{O}(n^{1/2 + \alpha})$ rounds.
For simplicity, we will begin by focusing on the first phase first; this phase works with the underlying graph $G$. Subsequent phases will be described separately because they work with the Schur complement graph and shortcut graph.  For now, we will also assume that every operation can be carried out with exact numerical precision.  We relax this requirement in Sections \ref{schurComplement} and \ref{numericalPrecision} and show that each operation can still be carried out with the requisite precision in the \cc.

\subsection{Phase 1}
Let $\rho := \lfloor \sqrt{n}\rfloor$.  We will build a random walk that terminates with high probability at the time $T$ where the walk first visits the $\rho$-th \textit{distinct} vertex.  
Using the fact that the cover time of any unweighted undirected graph is $O(n^3$) and Markov's inequality, we see that an $O(n^3)$ length random walk covers the graph with probability at least $\frac{1}{2}$.  Thus, $T = \Tilde{O}(n^3)$ with high probability.\footnote{Technically, by a result of Barnes and Feige \cite{FB}, in the first phase a walk only needs to be of expected length $O(n^{1.5})$ to visit $\sqrt{n}$ distinct vertices. However, this bound only applies to undirected, unweighted graphs. 
In subsequent phases, we consider walks on the Schur complement graph, which is weighted.
The Barnes and Feige result does not hold for weighted graphs. 
However, the $O(n^3)$ upper bound on the cover time remains a valid bound in later phases as well because the weighted graphs we work with are obtained by taking the Schur complement of the input (unweighted) graph.}
Initially, we begin generating a walk of \textit{target length} $\ell := \Tilde{\Theta}(n^3)$.  Specifically, we choose $\ell$ to be the smallest power of two at least $\log\big(\frac{4\sqrt{n}}{\epsilon}\big)n^3$.  We choose this length so that $\ell \geq T$ in every phase with probability at least $\frac{\epsilon}{2}$ by the union bound and because powers of 2 are convenient.  We will truncate this walk as needed so that the walk ends at time $T$.

\subsubsection{Sequential Random Walk Algorithm}
For clarity, we will start by describing a sequential version of the algorithm -- which contains some key ideas of our final algorithm -- and then show how to implement this efficiently in the \cc. 
The phase then begins with an Initialization Step, where the transition matrix of the random walk on the graph, $\adj$, as well as the matrix powers $\adj^2, \adj^4, \adj^8, ..., \adj^\ell$ are computed.

We now describe the ``top down'' walk filling process.
The initial partial walk $W_1$ is generated by first choosing an arbitrary start vertex $W[0]$.  The end vertex, $W[\ell]$, is sampled with the correct probability, conditioned on $W[0]$ being the start vertex. In particular, the distribution for $W[\ell]$ is given by $\adj^{\ell}[W[0],*]$.
The walk is then iteratively filled level-by-level in $O(\log \ell)$ \textit{levels}.  

Suppose the levels are labeled $1, 2, \ldots, \log_2 \ell$ and let $W_i$ denote the partial walk at the beginning of level $i$. Note that $W_i$ is a random walk on the graph with transition matrix $\adj^{\ell/2^{i-1}}$.
During level $i$, we construct the partial walk $W_{i+1}$ from $W_i$, 
by considering each pair of consecutive vertices in $W_i$ in chronological order, and
placing a vertex, which we will call a \textit{midpoint}, between this pair.
For example, at the beginning of level $1$, the initial partial walk is $W_1 = (W[0], W[\ell])$. During level $1$, the midpoint $W[\ell/2]$ is sampled and inserting this leads to the partial walk $W_2 = (W[0], W[\ell/2], W[\ell])$.
During any level, as described earlier in Section \ref{formulaPar}, to choose a midpoint between consecutive vertices $W[p]$ and $W[q]$, we sample a vertex $w$ with probability proportional to
\begin{equation}\label{midpointFormula}
    \adj^{(q-p)/2}[W[p], w] \cdot \adj^{(q-p)/2}[w, W[q]]
\end{equation}
Finally, we return $W_{\log_2(\ell) + 1}$.

\begin{lemma}
This sequential algorithm correctly samples a random walk of length $\ell$ beginning at $W[0]$.
\end{lemma}
\begin{proof}
    We can inductively show that $W_i$ is a correct random walk on the graph with transition matrix $\adj^{\ell/2^{i-1}}$.  The base case, $i=1$, is trivial.  The inductive step simply follows from the chain rule of probability and the Markov property of random walks.
\end{proof}

\subsubsection{Sequential Truncated Random Walk Algorithm}
Unfortunately, in the \cc{} model, we don't know how to efficiently simulate the previous algorithm once more than $\rho$ distinct vertices are contained in the random walk.  Note that with high probability a length $\ell$ walk will contain more than $\rho$ distinct vertices.  Instead, we let $\tau$ be the stopping time of the random walk on $G$ that is the minimum of $\ell$ and the time at the first occurrence of the $\rho$-th distinct vertex. Our goal is to sample random walks ending at time $\tau$.  The minimum of $\ell$ is included in the definition of $\tau$ to handle the low probability event where fewer than $\rho$ distinct vertices are contained in a random walk of length $\ell$.

We now augment the previous algorithm as follows.
Consider a level $i$, that starts with a partial walk $W_i$ of target length $\ell_i$. 
In our notation, the target length of a partial walk is given by the index of its final element.
Recall that $\ell_1$, the target length of $W_1$ (also denoted by $\ell$), is equal to $\Tilde{\Theta}(n^3)$.
In level $i$, we fill in midpoints in chronological order into the partial walk $W_i$ to construct the partial walk $W_{i+1}$.
As we insert a midpoint $w$ between consecutive vertices $W[p]$ and $W[q]$ in $W_i$, we check if the partial walk contains at least $\rho$ distinct vertices.  If so, we truncate the partial walk to end at the first occurrence of the $\rho$-th distinct vertex.
We call this truncated partial walk $W_{i+1}$, and the index of the 
first occurrence of the $\rho$-th distinct vertex is the new target length $\ell_{i+1}$. 
It is possible that no truncation occurs in a level and $W_{i+1}$ is obtained from $W_i$ after \textit{all} midpoints are filled in.  
In this case, $\ell_{i+1} = \ell_i$.
Again, we return $W_{\log_2(\ell) + 1}$.

\begin{lemma} The sequential truncated random walk algorithm samples a random walk beginning at $W[0]$ and ending at the stopping time $\tau$.
\end{lemma}
\begin{proof}
    By the Markov property of random walks, each truncation will have no impact on the probability of vertices placed before the truncation in future levels.  Thus, without changing the result returned by the algorithm, we could choose to defer all truncations of the walk to the end of the algorithm.  However, it is easy to see that this is simply generating a walk of length $\ell$ and truncating the walk to end at time $\tau$.
\end{proof}

\subsubsection{Truncated Walk in the Congested Clique}
\label{subsubsec: truncated}
Now we describe how to implement the sequential truncated random walk algorithm described above in the \cc{} model.

\medskip

\noindent
\textbf{Initialization Step.}
We begin with the same Initialization Step as in the sequential algorithm; see the pseudocode below (Algorithm \ref{markov-alg-preproc}). 
The purpose of the Initialization Step is twofold: (i) for each Machine $i$ to hold row $i$ and column $i$ in each of transition matrices $\adj, \adj^2, \adj^4, \ldots, \adj^\ell$ and (ii) to construct the initial partial walk subsequence $W_1 = (W[0], W[\ell])$.

\begin{algorithm}[H]
\caption{Initialization Step}
\label{markov-alg-preproc}
\begin{algorithmic}[1]
\State Machine $1$ is designated as the ``leader'' machine $\sfM$; $\sfM$ designates the vertex it is hosting (vertex 1) as $W[0]$, the start of the walk.
\State Using the \cc{} matrix multiplication algorithm from \cite{MatrixMult}, every Machine $i$ computes rows $\adj[i, *] ,\adj^2[i, *],\adj^4[i, *], \ldots,\adj^\ell[i, *]$.
\State Every Machine $i$ sends $\adj^k[i,j]$ to machine $j$, for all $j, k$.
\State  $\sfM$ samples $W[\ell]$ from the distribution given by $\adj^\ell[1,*]$.
\end{algorithmic}
\end{algorithm}
\noindent
We now describe level $i$ of the midpoint filling in process. 
We start with the inductive assumption that $\sfM$ holds partial walk $W_i$ of target length $\ell_i$ containing at most $\rho$ distinct vertices.

\medskip

\noindent
\textbf{Midpoint Requests.} 
$\sfM$ needs to generate midpoints between consecutive pairs of vertices in $W_i$.  By Formula \ref{midpointFormula}, the probability distribution of a midpoint vertex only depends on the identities of the vertices it is being placed between and the length of the walk between these vertices.  Thus, in a given level, all midpoints being placed between the same \textit{start} and \textit{end} vertices are drawn from the same probability distribution.  
Since $W_i$ contains at most $\sqrt{n}$ distinct vertices, $W_i$ can contain at most $n$ distinct (start, end) pairs.
$\sfM$ will designate one machines for each (start, end) pair that $\sfM$ needs to sample midpoints between.  Specifically, let the machine assigned the (start, end) pair $(p, q)$ be denoted $\sfM_{p,q}$. Furthermore, let $c_{p,q}$ denote the number of occurrences of the (start, end) pairs $(p, q)$ as a pair of consecutive vertices in $W_i$. 
The ``leader'' machine $\sfM$, requests $c_{p, q}$ midpoints from machine $\sfM_{p, q}$.

\medskip

\noindent
\textbf{Midpoint Generation.} We now consider a machine $\sfM_{p,q}$ receiving a request for $c_{p,q}$ midpoints. The first task of machine $\sfM_{p,q}$ is to acquire the distribution from which to sample the midpoints of (start, end) pair $(p, q)$.
Recall that consecutive vertices in $W_i$ are at distance $\ell/2^{i-1}$ in the underlying graph $G$; for convenience let $\delta$ denote $\ell/2^{i-1}$. For each vertex $j$, machine $\sfM_{p,q}$ sends a request to Machine $j$ (i.e., the machine hosting vertex $j$), requesting the (unnormalized) probability that vertex $j$ is the midpoint of a length $\delta$ walk between $p$ and $q$. Note that by Formula (\ref{midpointFormula}) this probability is just $\adj^{\delta/2}[p,j] \cdot \adj^{\delta/2}[j,q]$.  Due to the Initialization Step (Steps 2 and 3), Machine $j$ holds both $\adj^{\delta/2}[p,j]$ and $\adj^{\delta/2}[j,q]$ for all (start, end) pairs $(p, q)$ and can compute this product.  $\sfM_{p,q}$ then uses the implicit normalized probability distribution $\left(\adj^{\delta/2}[p,j] \cdot \adj^{\delta/2}[j,q]\right)_{j=1}^n$ it received from every other machine to independently sample a sequence of midpoints $\Pi_{p, q} = \pi_1, \ldots, \pi_{c_{p,q}}$, where $\pi_i$ is the midpoint intended to be inserted between the $i$th occurrence of the (start, end) pair $(p, q)$ in $W_i$. 

\begin{algorithm}[H]
\caption{Midpoint Requests and Generation}
\label{markov-alg-main}
\begin{algorithmic}[1]
\State Machine $\sfM$ holds the current partial walk $W_i$.
\State $\sfM$ counts the number of distinct consecutive pairs in $W_i$ and the number of occurrences $c_{p,q}$ of each pair $(p, q)$.
\State $\sfM$ assigns one machine, $\sfM_{p,q}$, to each (start, end) pair $(p, q)$ in $W_i$ and sends $\sfM_{p, q}$ the corresponding count $c_{p, q}$.
\State $\sfM_{p, q}$ requests and receives the probability $\adj^{\delta/2}[p,j]\cdot \adj^{\delta/2}[j, q]$ from each machine $j$ hosting vertex $j$.
\State Machine $\sfM_{p,q}$ samples the midpoint sequence $\Pi_{p,q} = \pi_1, \ldots ,\pi_{c_{p, q}}$, where each $\pi_k$ is sampled according to the (unnormalized) distribution $\left(\adj^{\delta/2}[p,j]\cdot \adj^{\delta/2}[j, q]\right)_{j=1}^n$ that machine $\sfM_{p,q}$ received.
\end{algorithmic}
\end{algorithm} 

\medskip

\noindent
\textbf{Distributed binary search for truncation point.} Finally, $\sfM$ needs to fill in the midpoints in level $i$, thereby constructing partial walk $W_{i+1}$ from $W_i$.  The challenge with this step is that the sequences $\Pi_{p, q}$ could be quite large, in fact even as large as $\Tilde{\Theta}(n^3)$! 
But, let us imagine for the moment that $\sfM$ ``magically'' receives $\Pi_{p, q}$ from each machine $\sfM_{p, q}$ and fills in the midpoints of each (start, end) pair $(p, q)$ in $W_i$ in the same order as specified by the sequence $\Pi_{p, q}$.  To faithfully execute the sequential truncated algorithm, we would obtain $W_{i+1}$ by truncating this filled in partial walk to end at index $\ell_{i+1}$.  Recall that $\ell_{i+1}$ is the first index of the $\rho$-th distinct vertex if it exists or otherwise $\ell_i$.

Returning to the algorithm, $\sfM$ is not immediately able to compute $\ell_{i+1}$ because it does not know the sequences $\Pi_{p,q}$. However, $\ell_{i+1}$ has already been determined by the partial walk $W_i$ and the sequences $\Pi_{p,q}$ for all (start, end) pairs $(p, q)$. 
Let $W_i^+$ denote the walk obtained by inserting the midpoints 
specified by the sequences $\Pi_{p,q}$ into $W_i$.
One of the key ideas of our algorithm is that even though no machine knows $W_i^+$, it is still possible for machine $\sfM$ to compute the truncation point $\ell_{i+1}$ efficiently in the \cc{} model by using distributed binary search. Once $\ell_{i+1}$ is found, $\sfM$ truncates the partial walk $W_i$ to end at time $\ell_{i+1}$.  Even after truncating $W_i$, the bandwidth does not exist for $\sfM$ to receive a truncated version of each sequence $\Pi_{p,q}$.  In the next section we will discuss how we place midpoints to form walk $W_{i+1}$.

To simplify the exposition, we first present a subroutine called $\chk$ that takes an integer $\ell' \leq \ell_i$ and returns \textsf{True} iff $\ell'\leq \ell_{i+1}$.  
This subroutine can be used to perform a binary search for the truncation point $\ell_{i+1}$.  Note that since $W_i^+$ is a partial walk, a given index in the walk may not hold a vertex.  We only perform the binary search over nonempty indices in the walk.  Let $W_i^+[0, \ell']$ denote the prefix of $W_i^+$ up to and including the vertex indexed $\ell'$.  There are two conditions that need to be checked: (i) whether $W_i^+[0, \ell']$ contains at most $\rho$ distinct vertices and (ii) if $W_i^+[0, \ell']$ contains exactly $\rho$ distinct vertices, whether the final vertex in $W_i^+[0, \ell']$ appears exactly once in the walk.
Denote the vertex at index $j$ in $W_i^+$ by $W_i^+[j]$.
Note that machine $\sfM$ can find $W_i^+[j]$ in $O(1)$ rounds.  
In particular, either the vertex at index $j$ is in $W_i$, in which case Machine $\sfM$ already knows this vertex or $W_i^+[j]$ is part of some $\Pi_{p, q}$ and $\sfM$ knows the pair $(p, q)$, in which case $\sfM$ can request $W_i^+[j]$ from the Machine $M_{p,q}$ holding that midpoint.  
We also let $c_{p,q}(\ell')$ denote the number of (start, end) pairs $(p, q)$ in $W_i$ truncated to end at index $\ell'$. 
We include a pair $(p, q)$ in this count even when $p$  
is in the truncated prefix of $W_i$, but $q$ falls past the index $\ell'$ in $W_i$.  Finally, we define $\Ct(p,q,j,\ell')$ as the number of occurrences of $j$ in $\Pi_{p,q}$ up to index $c_{p,q}(\ell')$.

\begin{algorithm}[H]
\caption{$\chk(\ell')$}
\label{check}
\begin{algorithmic}[1] 
    \State Machine $\sfM$ sends $c_{p, q}(\ell')$ to $\sfM_{p, q}$ for each (start, end) pair $(p, q)$ in $W_i$. 
    \State For every vertex $j$, machine $\sfM_{p, q}$ sends $\Ct(p,q, j, \ell')$ to Machine $j$.
    \State For every vertex $j$, Machine $j$ sends $\Ct(j,\ell') = \sum_{(p,q)} \Ct(p,q,j,\ell')$ to $\sfM$
    \State $\sfM$ computes \textsf{Dist}, the number of distinct vertices $j$ that are either contained in $W_i[0,\ell']$ or have $\Ct(j, \ell') > 0$. 
    \If{$\textsf{Dist} > \rho$}: \Return \textsf{False}
    \EndIf

    \State{$\sfM$ computes \textsf{CountLast}, the number of occurrences of $W_i^+[\ell']$ in $W_i[0,\ell']$ plus $\Ct(W_i^+[\ell'], \ell')$}

    \State \Return $(\textsf{Dist}<\rho)\lor (\textsf{CountLast} = 1)$
\end{algorithmic}
\end{algorithm}

\noindent
 \textbf{Midpoint placement via perfect matching sampling.}
 Truncating $W_i$ to end at $\ell_{i+1}$ implies that each sequence $\Pi_{p,q}$ can be truncated to end at index $c_{p,q}(\ell_{i+1})$. However, even this truncation does not suffice in reducing the size of each $\Pi_{p,q}$ to a point where the machines $M_{p,q}$ can communicate these sequences to Machine $\sfM$.

 The next key idea of our algorithm is that the midpoints need not be inserted in the order specified by $\Pi_{p, q}$. They can in fact be inserted in a randomly sampled order, as long as the chronologically final midpoint is accurately placed. This implies that instead of each machine $M_{p,q}$ sending the sequence $\Pi_{p,q}$ to Machine $\sfM$, it is enough for Machine $\sfM$ to know the final midpoint $m_f$ as well as the multiset, $\mathbb{M}\setminus\{m_f\}$, of the other midpoints to insert into $W_i$ to form $W_{i+1}$.  Specifically, the total multiplicity of a vertex $j$ in all the sequences $\Pi_{p,q}$ is given by $\Ct(j, t)$ and Machine $j$ knows this multiplicity and can send it to Machine $\sfM$.  Special care is taken for the final midpoint as placing it in a different position may violate condition (ii) from the previous section.

 In particular, we want to sample an ordering of the midpoints conditioned on the fact that $\mathbb{M}\setminus\{m_f\}$ is the multiset of the non-final midpoints and $m_f$ is placed last.  Recall that finding the final midpoint is easy as we can query $W_i^+[j]$ for any $j$.  Let $S$ be the set of vertices either contained in $W_i[0,\ell_{i+1}]$ or $\mathbb{M}$. Note that $|S| = O(\sqrt{n})$.  $\sfM$ first broadcasts $S$ in $O(1)$ rounds and then receives the submatrix $\adj^{\delta/2}_S$.  $\sfM$ can receive this submatrix in $O(1)$ rounds as it has size $O(n)$.

 After placing the final midpoint in the final midpoint position of $W_i$, we can think of placing the remaining midpoints as sampling a perfect matching in a weighted complete bipartite graph.  
 Specifically, let $\mathbb{M}' = \mathbb{M}\setminus\{m_f\}$ and $\mathbb{P}'$ denote the set of the remaining midpoint positions, i.e., excluding the final midpoint position. Note that multiple midpoint positions in $\mathbb{P}'$ could correspond to the same start-end pair.
 Now consider the complete bipartite graph $\mathcal{B} = K_{|\mathbb{M}'|, |\mathbb{P}'|}$ defined by the sets $\mathbb{M}'$ and $\mathbb{P}'$. To each edge $(x, y)$ in this graph, connecting a midpoint $x \in \mathbb{M}'$ to a midpoint position $y \in \mathbb{P}'$, we assign the weight $\adj^{\delta/2}[p, x]\adj^{\delta/2}[x,q]$ if midpoint position $y$ corresponds to the start-end pair $(p, q)$.

 As described in the introduction, we let the weight of a perfect matching in $\mathcal{B}$ be the product of all of its edge weights. We approximately sample a perfect matching of $\mathcal{B}$ proportional to its weight using the classical results of Jerrum, Sinclair, and Vigoda \cite{Permanent} and Jerrum, Valiant, and Vazirani \cite{JVV}.
 We choose the total variation distance error of the sampler to be at most $\frac{\epsilon}{4 \sqrt{n}\log \ell}$, so that after taking a union bound across all levels and phases the total variation distance error becomes at most $O(\epsilon)$.  This perfect matching assigns to each midpoint a position in the walk.  This completes the description of the first phase of our \cc{} adaptation of the truncated random walk algorithm.
\begin{lemma}\label{matchingLemma}Consider the partial walks $W_i$ and $W_{i+1}$ after level $i$ in the sequential truncated random walk algorithm.  Let $\mathbb{M}$ be the multiset of new midpoints in $W_{i+1}$.  Let $m_f$ be the final such midpoint.  Let $\sigma$ be the total chronological sequence of newly added midpoints to $W_{i+1}$, excluding the final midpoint.  $P(\sigma | \mathbb{M}, m_f)$ is proportional to the weight of any perfect matching in $\mathcal{B}$ that assigns non-final midpoints in the same order as $\sigma$.
\end{lemma}
\begin{proof}
    First, note that in the trivial case where $\sigma$ is not a permutation of $\mathbb{M}\setminus\{m_f\}$, $P(\sigma | \mathbb{M}, m_f) = 0$ and there are no corresponding perfect matchings.  Otherwise, $P(\sigma | \mathbb{M}, m_f) \propto P(\sigma \land \mathbb{M} \land m_f) = P(\sigma \land m_f)$.  Let the midpoints be chronologically labeled $\sigma_1,\sigma_2,...,m_f$.  By the Markov property, since vertices already exist in the walk separating each midpoint position, each midpoint is sampled independently.  Specifically, let $p_i$ be the probability that $\sigma_i$ is chosen as a midpoint between the $i$-th start/end pair in $W_i$.  This probability is given, unnormalized, by Equation \ref{midpointFormula}.  Then 
    $$P(\sigma | \mathbb{M}, m_f) \propto P(\sigma \land m_f) \propto \prod_i p_i$$
    Note that by definition this is proportional to the weight of any corresponding perfect matching.
\end{proof}

\begin{lemma}\label{couplingLemma} The algorithm above generates a random walk beginning at vertex $W[0]$ and ending at vertex $W[\tau]$ with total variation distance within $\frac{\epsilon}{4\sqrt{n}}$.
\end{lemma}
\begin{proof}Think of the sequential truncated random walk algorithm which we already know to be a perfect sampler.  Placing the midpoints in walk $W_i$ to form walk $W_{i+1}$ can be thought about in three theoretical steps, even though this is not the order the algorithm performs them.  First, the final midpoint, $m_f$, is sampled and placed in $W_i$.  Next, a multiset of the remaining midpoints $\mathbb{M}\setminus\{m_f\}$ is sampled.  Finally, a permutation of $\mathbb{M}\setminus\{m_f\}$, $\sigma$, is chosen for placing the remaining midpoints in $W_i$.  Note that our distributed algorithm chooses $m_f$ and $\mathbb{M}\setminus\{m_f\}$ with the exact correct probability, as $W_i^+$ is a true random walk even if it cannot be collected at machine $\sfM$.  Further, by Lemma \ref{matchingLemma} and noting that all permutations have the same number of corresponding perfect matchings in $\mathcal{B}$, the midpoints are also placed with the correct probability.  Thus if we sampled the perfect matching with no error, the walk would be sampled with no error.

By the union bound, the error in our algorithm is at most the same error with which we sample perfect matchings multiplied by the number of perfect matchings we sample.  We sample $O(\log \ell)$ perfect matchings as there are $O(\log \ell)$ levels in a phase.  To be more precise this follows using probabilistic coupling.  We can think of a sampler with error $\psi$ as returning the same sample as some true sampler with probability $1-\psi$.  Furthermore, the error in our algorithm can be bounded by the probability that at any point the approximate sampler for perfect matchings differs from the true uniform sampler.
\end{proof}

Lastly, we analyze the runtime for the first phase of the algorithm.  Note that the time is dominated by the matrix multiplication.

\begin{lemma}\label{timePerPhase}
    The first phase of the algorithm requires $\Tilde{O}(n^\alpha)$ rounds.
\end{lemma}
\begin{proof}
Recall that $\ell$ is $\Tilde{\Theta}(n^3)$. Note that the initial matrix exponentiation can be done, by successively powering the transition matrix, in $\Tilde{O}(n^{\alpha})$ rounds.  Each machine receiving its column in each of the matrices then takes $\Tilde{O}(1)$ rounds.  

It is then sufficient to argue that the $O(\log n)$ levels each require $O(\log n)$ rounds.  Each communication we describe between $\sfM$ and the other machines takes $O(1)$ rounds.  The only non-trivial step here is $\sfM$ broadcasting the set $S$ and receiving back the corresponding submatrix.  As $|S| = O(\sqrt{n})$, broadcast can be done in two rounds.  Furthermore, as the submatrix has size $O(n)$, $\sfM$ can receive it in $O(1)$ rounds using Lenzen's protocol.  Finally, the binary search process takes a total of $O(\log n)$ rounds.
\end{proof}

\subsection{Subsequent Phases}\label{futurePhases}
We now describe a phase of the algorithm after the initial phase.  
These later phases will make use of the derivative graphs \schur$(G,S)$ and \shortcut$(G,S)$. In this subsection we assume that the transition matrices of these graphs have already been computed and distributed, with Machine $i$ holding row $i$ and column $i$ of each of these matrices.
We describe how to compute these transition matrices in Section \ref{schurComplement}.

As in Phase 1, we need to sample a random walk visiting $\rho$ distinct vertices.  This walk will serve as a continuation of the walk generated in the previous phases; however, we avoid revisiting vertices from earlier phases (except the final vertex from the previous phase).  Let $v_f$ be the final vertex visited by the walk of the previous phase.  
Let $\old$ denote the set of vertices visited by the final walk of any previous phase.
Let $S = \{v_f\} \cup (V \setminus \old)$ denote the as-yet-unvisited vertices, along with $v_f$, which will serve as the starting point of the walk in this phase.
Our goal in this and subsequent phases is to find the first visit edge for all vertices in $S$ except $v_f$.  

We use the same algorithm as we used for Phase 1, while making sure that we skip over the vertices in $\old$.  We need to skip vertices in $\old$ to ensure we have the network bandwidth to visit $\rho-1$ distinct new vertices.  In order to do this, we repeat the algorithm for Phase 1 on the Schur complement of $G$ with respect to $S$, \schur$(G,S)$ (see Section \ref{section:Preliminaries}), instead of $G$.  This
explicitly removes all vertices in $\old \setminus \{v_f\}$ from the graph, while ensuring that we are taking a random walk that is equivalent to taking a random walk on $G$ skipping over vertices in $\old \setminus \{v_f\}$.
Since the cover time of \schur$(G,S)$, for any $S \subseteq V$, is bounded above by the cover time of $G$, we can (as in Phase 1) start with a partial walk $W_1 = (W[0], W[\ell])$ with target length 
$\Tilde{\Theta}(n^3)$.

However, taking a random walk on \schur$(G,S)$ does not immediately provide first visit edges in $G$ for the vertices in $S$ visited by the walk.
This requires additional work and in particular involves the use of the shortcut graph \shortcut$(G, S)$ (see Section \ref{section:Preliminaries}). We assume each vertex holds both its row and column in the transition matrix of \shortcut$(G, S)$.
Suppose that $W[i]$ is the first appearance of some vertex $v$ in the random walk.  Since the walk in  
\schur$(G,S)$ corresponds to taking ``shortcuts'' in the original graph (skipping over vertices in $\old \setminus \{v_f\}$), we cannot say that the first entrance edge for $v$ is $(W[i-1], W[i])$.  

We now describe our algorithm for sampling the first incoming edge for $v$ given that the last vertex visited in $S$ was $W[i-1]$.  Let $u$ be a neighbor of $v$ in $G$ and $\deg_S(u)$ be the number of neighbors of $u$ in the set $S$ in the graph $G$.  Then, for all neighbors $u$ of $v$ in $G$, an edge $(u,v)$ is sampled as the first visit edge to $v$ with probability proportional to $\sct[W[i-1], u] \cdot \frac{1}{\deg_S(u)}$. Here $\sct$ is the random walk transition matrix of \shortcut$(G, S)$.  The correctness of this expression follows by Bayes' rule, as $\sct[W[i-1], u]$ gives the probability that a walk started at $W[i-1]$ in $G$ visits $u$ immediately before its first visit (at a time greater than zero) to a vertex in $S$ and $\frac{1}{\deg_S(u)}$ gives the probability that the vertex in $S$ that is then visited is $v$.

Finding the first visit edges of all vertices visited in this phase, via the sampling method
described above, can be implemented in $O(1)$ rounds in the \cc{} model.  See the following psuedocode.

\begin{algorithm}[H]
\caption{Sample first visit edges}
\begin{algorithmic}[1]
    \State Machine \sfM{} obtains walk $W$ on \schur$(G,S)$
    \For{Each distinct vertex $v\neq v_f \in W$}
        \State Machine \sfM{} finds $i = \min\{j\,|\,W[j] = v\}$
        \State Machine \sfM{} sends $W[i-1]$ to machine $v$
        \For{Each neighbor $u$ of $v$}
            \State Machine $v$ requests $\sct[W[i-1], u] \cdot \frac{1}{\deg_S(u)}$ from machine $u$
        \EndFor

        \State Machine $v$ samples vertex $v$'s first entrance edge $(u,v)$, where $u$ is sampled from the (unnormalized) distribution $\left(\sct[W[i-1], u] \cdot \frac{1}{\deg_S(u)}\right)_{u\in N(v)}$ 

    \EndFor
\end{algorithmic}
\end{algorithm}

\subsection{Analysis}
\begin{lemma}\label{mainAlgCorrectness}The final spanning tree is drawn with total variation distance at most $\epsilon$ from the uniform distribution. 
\end{lemma}
\begin{proof}
    Suppose for the moment that in each phase we generate a truly random walk on the graph \schur$(G,S)$ that ends at the first visit to the $\rho$-th distinct vertex.  It is easy to see, by the correctness of the Aldous-Broder algorithm, that we would generate a spanning tree of $G$ sampled uniformly at random.

    Now, recall that we chose $\ell$ such that the probability that any non-final phase fails to visit $\rho$ distinct vertices is at most $\epsilon/2$.  Say in this case we simply output an arbitrary spanning tree.  Furthermore, in each phase, the perfect matching we sample has total variation distance error of at most $\frac{\epsilon}{4\sqrt{n}}$.  We note that there are at most $2\sqrt{n}$ phases for sufficiently large $n$, as we visit at least $\sqrt{n}-2$ distrinct new vertices in each phase.  The total variation distance error of our algorithm then follows by again applying the union bound along with a coupling argument (as in Lemma \ref{couplingLemma}) across all phases and adding the two sources of error.
\end{proof}

\subsection{Computing Schur Complement and Shortcut Graphs}
\label{schurComplement}
We now demonstrate how to approximately compute the shortcut and Schur complement graphs.  In this section we use the term \textit{subtractive error} to refer to negative additive error to emphasize that each approximation is an under-approximation.  We first give a lemma regarding the accuraracy to which we can carry out matrix multiplication.

\begin{lemma}\label{MatrixRoundingLemma}Suppose $M$ is a $n\times n$ transition matrix.  Let $k$ be a power of $2$ that is $O(n^{c_1})$ for $c_1 > 0$.  Let $\beta$ be $\Omega(1/n^{c_2})$ for $c_2 > 0$.  In the \cc{} we can compute $M^k$ with subtractive error at most $\beta$ in $\Tilde{O}(n^{\alpha})$ rounds.
\end{lemma}
\begin{proof}
    Let $M'$ be \textit{round}$(M)$, where \textit{round} returns the input matrix with each entry truncated after $O(\log\frac{1}{\delta})$ bits.  In particular, each entry of \textit{round}$(M)$ has at most $\delta$ subtractive error over $M$.  
    
    We will now define $M'(k)$ inductively.  Let $M'(1) = M'$.  Further, when $k$ is a power of two, let $M'(k) = \textit{round}([M'(k/2)]^2)$.  We now analyze the subtractive error, $E(k)$, between $M^k$ and $M'(k)$.  Since the sum of any row is at most $1$ and the sum of any column is at most $n$, $E(1) \leq \delta$ and $E(k) \leq (n+1)E(k/2)+\delta$.  Thus $E(k)$ is $O(\delta k^{c\log k})$ for some $c>0$.

    The proof now follows by choosing $\delta = \Theta(\frac{\beta}{k^{c\log k}})$ and computing $M'(k)$.  Using the \cc{} matrix multiplication algorithm and the fact that each matrix entry requires $O(\log\frac{1}{\delta}) = O(\log^2(n))$ bits we can compute $M'(k)$ in $\Tilde{O}(n^{\alpha})$ rounds.
\end{proof}

\begin{corollary}Let $S$ be any subset of the vertices of $G$.  Let $\beta$ be $\Omega(1/n^{c})$ for $c > 0$.  In the \cc{} we can compute $\sct$, the transition matrix for a random walk on \shortcut$(G,S)$, with subtractive error at most $\beta$ in $\Tilde{O}(n^{\alpha})$ rounds.
\end{corollary}
\begin{proof}
    Let the transition matrix of the random walk on $G$ be $P$.  We will first construct an auxiliary graph $G'$.  The vertices of $G'$ consist of $L \cup R$, where both $L$ and $R$ contain a new copy of every vertex in $V$.  For a vertex $v \in V$, call its copy in $L$, $v'$, and call its copy in $R$, $v''$.  Again we define $G'$ by defining the transition probabilities, $R$, of its random walk.

    $$\begin{cases}
        R[u'', u''] = 1\\
        R[u', v'] = P[u,v],\,\text{if}\, v'\notin S\\
        R[u', u''] = \sum_{v\in S} P[u,v]
    \end{cases}$$

    Now consider $R^\infty = \lim_{k\to\infty} R^k$.  We can see that $\sct[u,v] = R^\infty[u', v'']$.  Note that since the cover time of $G$ is $O(n^3)$, we can get a $\delta$ subtractive approximation of $R^\infty$ by choosing $k = O(n^3\log\frac{1}{\delta})$.  This follows because all the vertices in $R$ are absorbing and in constant units of cover time, with probability at least $1/2$, any vertex in $L$ will have a transition to a vertex in $R$.   Now the result follows by combining the approximation of $R^\infty$ and Lemma \ref{MatrixRoundingLemma} using the triangle inequality.
\end{proof}

\begin{corollary}\label{shortcut-construction}Let $S$ be any subset of the vertices of $G$.  Let $\beta$ be $\Omega(1/n^{c_1})$ for $c_1 > 0$.  In the \cc{} we can compute $\scc$, the transition matrix for a random walk on \schur$(G,S)$, with subtractive error at most $\beta$ in $\Tilde{O}(n^{\alpha})$ rounds.  Furthermore, this is also true for $\scc^k$, where $k$ is a power of $2$ and $k$ is $O(n^{c_2})$ for $c_2 > 0$.
\end{corollary}
\begin{proof}
    Let $\sct $ be the transition matrix of \shortcut$(G,S)$.  Let $R$ be a transition matrix defined over $V$ as follows.

    $$R[u,v] = \begin{cases}
        1,\,\text{if}\, u = v \land \deg_S(u) = 0\\
        \frac{1}{\deg_S(u)},\,\text{if}\, \{u,v\} \in E \land v \in S\\
        0,\,\text{otherwise}
    \end{cases}$$

    Recall that unlike $\sct$ and $R$, $\scc$ is only defined over vertices in $S$.  Now $\scc[u,u] = 0$, for all $u$.  Otherwise, $\scc[u,v]$ is proportional to $(\sct R)[u,v]$, with a different proportionality constant for each row.  These constants arise because we don't allow self transitions in $\scc$.   In particular, we need to multiply each row $u$ by $M_u = \frac{1}{1-(\sct R)[u,u]}$.  Note that $(\sct R)[u,u]$ is the probability that a random walk in $G$ started at $u$ revisits $u$ before visiting any other vertex in $S$.  $M_u$ is bounded by a polynomial, $O(n^{c})$ for some $c>0$, by Proposition 2.3 of \cite{CommuteBound}.  
    
    The result now follows by simple algebra and the previous corollary and lemma.
\end{proof}

\subsection{Numerical Precision}
\label{numericalPrecision}
Finally, we show that the entire algorithm can be carried out with the requisite numerical precision.  We analyze the algorithm using approximate probabilities instead of exact probabilities.  We will assume that each matrix has been computed with subtractive error at most 
$\beta$.  $\beta$ will be chosen later in the proof and will be large enough such that each entry will fit in $O(1)$, $O(\log n)$ bit, words of the \cc{} model.  Specifically, $\beta$ can be made $\Omega\big(\frac{1}{n^c}\big)$ for some $c>0$.  All other computations can then be carried out with full precision.  We first prove one helper lemma.

\begin{lemma}
    \label{MidpointTVDLemma}
    Suppose the distribution for placing midpoint $m$ between vertices $p$ and $q$ is given by $\frac{\mathcal{W}[p,m]\mathcal{W}[m,q]}{\mathcal{W}^2[p,q]}$.  Further, suppose $\mathcal{W}^2[p,q] > \frac{1}{n^c}$, for some fixed $c>0$.  We can compute this distribution with maximum total variation distance error $2\beta n^{c+1}$ assuming we have a $\beta$ subtractive approximation of $\mathcal{W}[p,*]$ and $\mathcal{W}[*,q]$.
\end{lemma}
\begin{proof}
    Let the approximation of matrix $\mathcal{W}$ be $\hat{\mathcal{W}}$.
    We will let the unnormalized approximate distribution, $\hat{U}$, be defined by $\hat{U}(m) = \mathcal{W}[p,m]\mathcal{W}[m,q]$.  Note that each term in the approximate distribution has subtractive error at most $2\beta$.  We define the exact normalized distribution $U$ in the same manner using exact probabilities.  Let $C = U - \hat{U}$.

    Note that we can couple the distributions $(U,\hat{U})$ in the following manner.  Sample $x$ from $\hat{U}$ and $y$ from $C$.  Let the sum of all elements in $\hat{U}$ be $Z$.  We choose $r\in[0,\mathcal{W}^2[p,q]]$ uniformly at random.  If $r < Z$, we return the coupled sample $(x,x)$.  Otherwise, return $(y,x)$.  It is easy to see this is a valid coupling, in particular the first component is drawn from the first distribution and the second component is drawn from the second distribution.  The total variation distance between the normalized distributions is bounded by the probability that the first element doesn't equal the second element in the coupling.  Thus the total variation distance is at most $P(r > Z)$.  This is bounded by $\frac{2\beta n}{\mathcal{W}^2[p,q]} \leq 2\beta n^{c+1}$.  Thus normalizing $\hat{U}$ gives our desired distribution.
\end{proof}

\begin{lemma}\label{BetaLemma}There exists a choice of $\beta$ such that the approximate algorithm draws spanning trees from a distribution with a total variation distance at most $\epsilon$ from the uniform distribution.
\end{lemma}
\begin{proof}
We will consider four versions of the subroutine for sampling random walks in our algorithm for sampling spanning trees.  We can use either the sequential truncated random walk algorithm or the distributed truncated random walk algorithm.  Either of these algorithms can be used with exact or approximate probabilities.  Using the the FPRAS of \cite{Permanent}, we know that we can make the distributions of spanning trees generated using the approximate sequential and approximate distributed random walk algorithms have a total variation distance of at most $\frac{\epsilon}{3}$.  Furthermore, we already know that generating spanning trees using the sequential truncated walk algorithm with exact probabilities can give a sample with total variation distance at most $\frac{\epsilon}{3}$.  Thus, by two applications of the triangle inequality, it is sufficient to show that for sufficiently small $\beta$, the distributions of trees generated using the sequential truncated random walk algorithms with exact and approximate probabilities also have a total variation distance of at most $\frac{\epsilon}{3}$.

We will use a coupling argument to show that the sequential algorithm using approximate probabilities only has a total variation distance of at most $\frac{\epsilon}{3}$ from the sequential algorithm using exact probabilities.  It follows easily from probabilistic coupling, that the total variation distance between the two algorithms outputs is at most the probability that the algorithms make a different decision at some point in their computation, using the same source of randomness.  Furthermore, it follows that if the two algorithms are drawing samples from distributions with total variation distance at most $\delta$, they will draw different samples with probability at most $\delta$.

Now inductively suppose that the exact and approximate algorithm have made the same choices so far.  Suppose we are sampling a midpoint $b$ between vertices $a$ and $c$.  We want to sample $b$ with probability $\frac{P[a,b]P[b,c]}{P^2[a,c]}$ for one of the matrices $P$.  For now assume that $P^2[a,c] > \frac{1}{n^{k_1}}$, for a yet to be chosen $k_1 > 0$.  By Lemma \ref{MidpointTVDLemma}, we can compute this distribution with total variation distance at most $2\beta n^{k_1+1}$.
This bound on the total variation distance is also a bound on the probability that the two algorithms will sample different choices of $b$.

Now, as we are choosing $\Tilde{O}(n^3)$ midpoints, we can take the union bound to get the probability that the algorithms ever sample a different midpoint is at most $$\tilde{O}\bigg(n^3 \beta n^{k_1+1}\bigg) + \Tilde{O}\bigg(n^3\frac{n}{n^{k_1}}\bigg)$$
The second term is a bound on the probability that our assumption on $k_1$ ever fails.  We also have to ensure the algorithms both sample the initial endpoint identically, but this is trivial.

We also have to bound the probability that the first visit edges are sampled differently.  We get an expression of a similar form, the only nontrivial fact that we need is that $M_u$ is bounded by some polynomial $n^{k_2}$, as seen in Corollary \ref{shortcut-construction}.  Then we get the result we need by choosing $k_1$ to be sufficiently large and $\beta$ to be sufficiently small.
\end{proof}

\subsection{Putting It All Together}
\mainresult*
\begin{proof}
   Altogether, the approximate algorithm has a runtime of $\Tilde{O}(n^{1/2 + \alpha})$ rounds as there are $O(\sqrt{n})$ phases which each take $\Tilde{O}(n^\alpha)$ rounds.  We already know that excluding computing approximate matrix powers, the Schur complement graph, and the shortcut graph, the time per phase is $\Tilde{O}(1)$ rounds, see Lemma \ref{mainAlgCorrectness}.  Note that Section \ref{schurComplement} along with the choice of $\beta$ from Lemma \ref{BetaLemma} implies that computing these takes $\Tilde{O}(n^\alpha)$ rounds per phase.
\end{proof}

\section{Fast Random Walk Via Doubling}
Bahmani, Chakrabarti, and Xin \cite{BCX} present an algorithm called \dbling that computes a random walk of length $\tau$ in $O(\log \tau)$ iterations in the Map-Reduce model.
In this paper, the Map-Reduce model is specified somewhat abstractly, without an explicit notion of machines, communication bandwidth, etc.
As a result, a single iteration (consisting of a Map phase, Shuffle phase, and Reduce phase) in this Map-Reduce model may require $\Omega(n)$ rounds in the \cc{} model and so using the \dbling algorithm directly in the \cc{} model is not efficient. 
In this section we present a ``load balanced'' version of the \dbling algorithm that can compute a length-$n$ walk in $\pln$ rounds in the \cc{} model.

\noindent
\textbf{High-level idea.} To compute a length-$\tau$ random walk, the \dbling algorithm starts with each vertex $v$ computing $\tau$ length-1 random walks (i.e., random edges) originating at $v$. 
During each iteration, half the walks held by a vertex are designated to serve as prefixes and the other half are designated to serve as suffixes in the merging process. 
In each iteration, walks that end at a vertex $v$ are merged with walks originating at $v$. 
After the merging process, each vertex $v$ continues to hold walks originating at $v$.
Note that this halves the number of walks held by a vertex, while doubling the length of every walk. After $O(\log \tau)$ of these ``doubling'' iterations, every node $v$ holds a length-$\tau$ random walk originating at vertex $v$. While each walk is a proper random walk, 
because of how the walks are merged (described in detail below), walks originating at different vertices are not independent.  
The main challenge for the \dbling algorithm is to implement the merging step efficiently in the \cc{} model. 
It turns out that even for relatively short walks, i.e., with $\tau = \Theta(n)$, a faithful implementation of the \dbling algorithm requires $\Omega(n^2 \log n)$ bits to travel to a particular vertex in a single merging step in the worst case, which 
requires $\Omega(n)$ rounds in the \cc{} model. 
We add a ``load balancing'' component to each merging step so as to take advantage of the 
overall $\Theta(n^2)$ bandwidth of the \cc{} model.
We show that with this ``load balancing'' step in place, each merging step takes $O(\log n)$ rounds, w.h.p.

\smallskip

\noindent
\textbf{Load-balanced Doubling Algorithm.} We now describe our algorithm.
Consider an iteration of the \dbling algorithm. We assume that just before the start of this iteration, each vertex holds a sequence of $k$ walks of length $\eta$ each. 
Let $W_v^1, W_v^2, \ldots, W_v^k$ denote the sequence of $k$ length-$\eta$ walks held by vertex $v$. For any $v$ and $i$, we use $W_v^i[{\sf end}]$ to denote the ID of the last vertex of the walk, $W_v^i$. 

We further assume that both $k$ and $\eta$ are powers of 2 and $k \cdot \eta$ is the smallest power of 2 that is at least $\tau$. Before the first iteration, $k$ is the smallest power of 2 that is at least $\tau$ and $\eta$ is 1. In each iteration, $k$ halves and $\eta$ doubles.

\begin{enumerate}
\item Machine $1$ picks a binary string $s$ of length $O( \log^2 n)$ uniformly at random and broadcasts $s$ to all other machines.  Every machine then uses $s$ to pick a hash function $h_s$ from a family of $8c\log n$-wise independent hash functions $\mathcal{H} = \{h : [n] \times [k] \to [n]\}$ for constant $c > 1$.\footnote{For positive integers $N > M$ and $t$, there is a family of $t$-wise hash functions, $\calH =\{h_s: [N]\rightarrow [M] \}$ that allows us to uniformly sample from $\calH$ using $O(t\cdot \log N)$ bits and $\text{poly}(\log N, t)$ time (sequential) computation \cite{Vadhan12}.} 


\item For $i=1, \dots, k/2:$ each machine $v$ sends the tuple $(v, i, W_v^i)$ to machine $v'= h_s(W_v^i[{\sf end}], k-i+1)$.

\item For $i=k/2+1,\dots, k:$ each machine $v$ sends the tuple $(v, i, W_v^i)$ to machine $v''=h_s(\id_v, i)$.



\item Each machine $w$ receiving two walks $W_u^i$ and $W_v^j$, where $W_u^i[{\sf end}] = \id_v$, $1 \le i \le k/2$, and $i + j = k+1$ concatenates $W_u^i$ and $W_v^j$ and sends the tuple $(i, W_u^i \circ W_v^j)$ to machine $u$. (Here $W_u^i \circ W_v^j$ denotes the concatenation of $W_u^i$ and $W_v^j$.)

\item Each machine $v$, on receiving a tuple $(i, W)$, sets $W_v^i := W$.
\end{enumerate}
Note that in the above index-based merging scheme (which is due to Bahmani, Chakrabarti, and Xin \cite{BCX}) a walk $W_u^i$ for $1 \le i \le k/2$, that ends at vertex $v$, is merged with a walk, $W_v^{k-i+1}$, originating at $v$. In other words, a walk from the first half of $u$'s sequence of walks is merged with the corresponding walk from the second half of $v$'s sequence of walks. As shown in \cite{BCX}, this index-based merging scheme suffices to guarantee that every walk is indeed a random walk.

For our analysis, we will need the following concentration inequality from Bellare and Rompel\cite{BR94}. 
\begin{fact}[$t$-wise Concentration Bound]
\label{fact:twise}
Let $Y_1,\dots, Y_m$ be $t$-wise independent random variables taking values in $[0,1]$ and $Y=\sum^m_{i=1} Y_i$ with $\E[Y]=\mu$. Then, for any $a>0$, we have, 
\begin{align*}
    \Pr\big[|Y-\mu|\geq a \big] \leq 8\bigg(\frac{t\mu+t^2}{a^2}\bigg)^{\frac{t}{2}}. 
\end{align*}
\end{fact}
\noindent
Equipped with the fact above, we prove the following key claim that roughly says that in the load-balanced doubling algorithm each vertex sends and receives $O(k \log n)$ tuples. 


\begin{lemma}
\label{clm:doubling}
    Let $c>1$ be any constant. Then, in step $2$ and step $3$, in the above algorithm, for any $v\in V$, it holds that:
    \begin{align*}
\Pr[\text{Machine }v \text{ receives } \geq  16ck\log n \text{ tuples }]\leq n^{-2c}
    \end{align*}. 
\end{lemma}
\begin{proof}
 In step $2$ and $3$ of the load-balanced algorithm above, a total of at most $nk$ number of tuples are distributed using a $t$-wise hash function among $n$ number of vertices where $t=8c\log n$. 
Towards proving the claim, fix any machine $v$. We define $Y_{j}$ to be the event that machine $v$ gets the $j$-th tuple for $j=1,2,\dots,nk$. Note that, $\mu:=\E\big[\sum Y_j\big]=nk\cdot \frac{1}{n}=k$. 
Now we set, $a= k\cdot t = 8kc\log n$. Using Fact~\ref{fact:twise} on these random variables, we get,
    \begin{align*}
        \Pr[v \text{ receives } \geq  16ck\log n \text{ tuples }]&\leq\Pr\big[|Y-k|\geq k\cdot t \big] \\& \leq 8\bigg(\frac{t(k+t)}{t^2k^2}\bigg)^{\frac{t}{2}}
        \\&\leq 8 \bigg(\frac{1}{8ck\log n}+\frac{1}{k^2} \bigg)^{4c\log n}\\&\leq 8\cdot k^{-4c\log n}\\&\leq 8\cdot 2^{-4c\log n}= n^{-2c}
    \end{align*}
In the fourth line, we use the crude bound $\bigg(\frac{1}{8ck\log n}+\frac{1}{k^2} \bigg)\leq \frac{1}{k}$ and in the last line, we used the assumption $k\geq 2$. 

 \end{proof}
    
Now, we are ready to bound the running time of one iteration of the load-balanced doubling algorithm. 
\begin{lemma}
A single iteration of the Load-balanced Doubling algorithm runs in 
\begin{itemize}
\item $O\left(\frac{\tau}{n} \log n\right)$ rounds with high probability, if $\tau = \Omega(n/\log n)$.
\item $O(1)$ rounds with high probability, if $\tau = O(n/\log n)$.
\end{itemize}
\end{lemma}
\begin{proof}
First, note that, each tuple of the form $(\id_v, i, W_v^i)$, used in our distributed doubling algorithm, requires $O(\eta\log n)$ bits to encode, which is $O(\eta)$ messages. 
We prove the claim by explicitly computing the number of bits communicated by each machine in each of the step in the algorithm.  
\begin{itemize}
\item Step 1: Machine 1 sends a randomly sampled string $s$ of size $O(\log^2 n)$ bits to all the other vertices. This can be completed in 2 rounds of communication.
    \item Step 2 and 3: Each machine $v$ sends $O(k)$ tuples; so sends $O(k\eta)$ messages.
    On the reception side, it follows from Lemma~\ref{clm:doubling} that each machine $v\in V$ receives less than $16ck\log n$ tuples with probability $\geq 1- n^{-2c}$. Thus, each machine receives $O(k \eta \log n)$ messages with high probability.
    Using Lenzen's routing protocol \cite{Lenzen13}, this communication can be completed in the \cc{} model in $O\left(\max\left\{\frac{k\eta}{n} \log n, 1\right\}\right)$ rounds. 

    \item Step 4: Each machine can only send at most the number of messages they received at the end of the Steps 2 and 3, which is $O(k\eta \log n)$ with high probability. On the other hand, in this step, each machine $v$ receives $k/2$ merged walks of $2\eta$-length. 
    It follows that in Step 4, each machine $v$, receives $O(k\eta)$ messages. 
    Thus, as in Steps 2 and 3, the communication in Step 4 can be completed in the \cc{} model in $O\left(\max\left\{\frac{k\eta}{n} \log n, 1\right\}\right)$ rounds.
\end{itemize}
Noting that in all iterations of the algorithm, $k \eta = \Theta(\tau)$, we obtain the result.
\end{proof}

The following theorem is an immediate consequence of the previous lemma, along with the observation that it takes the \dbling{} algorithm $O(\log \tau)$ iterations to go from length-1 walks to length-$\tau$ walks.
\secondResult*

\section{Conclusion}
Sampling random spanning trees uniformly is a fundamental problem with beautiful mathematical connections to random walks, electrical circuits, and graph Laplacians. The problem also has important applications to graph sparsifiers and serves as the basis for new approximation algorithms.
Using the Aldous-Broder algorithm based on random walks, we present the first sublinear round algorithm for approximately sampling uniform spanning trees in the \cc{} model of distributed computing. 
In particular, our algorithm requires $\Tilde{O}(n^{1/2 + \alpha})$ rounds for sampling a spanning tree from a distribution with total variation distance within $1/n^c$, for arbitrary constant $c > 0$, from the uniform distribution.  In addition, we show how to take somewhat shorter random walks, even more efficiently in the \cc{} model, leading to much faster spanning tree sampling algorithms for graphs with small cover times.

We hope that our results and techniques will spark new research on distributed sampling of spanning trees. In the \cc{} model, we have identified some potential barriers for improving the $\tilde{O}(n^{1/2+\alpha})$ running time using our approach. New ideas might be needed to overcome these barriers.
A completely different direction involves implementing the MCMC approach, e.g., using the down-up walk from \cite{UpDownWalk}, efficiently in the \cc{} model. The problem is also poorly understood in other models of distributed computing such as the MPC and \congest{} models.

\paragraph{Acknowledgments.} 
Sriram Pemmaraju was supported in part by the National Science Foundation (NSF) grant CCF-2402835.

\bibliographystyle{plain}
\bibliography{main}
\newpage
\section{Appendix: Exactly Sampling Spanning Trees}
\label{Exact Sampling Appendix}
Our main spanning tree algorithm as stated has three sources of error.  We show that the algorithm can be made to sample exactly from the uniform distribution with no error.  This carries a tradeoff in time, taking $\Tilde{O}(n^{2/3+\alpha}) = O(n^{0.824})$ rounds.  It is worth noting that for most practical uses, inverse polynomial error is sufficient, as even distinguishing between exact and approximate sampling will take a superpolynomial number of samples.

\begin{enumerate}
    \item The first source of error arises because there may be a phase where a walk fails to visit $\Omega(\sqrt{n})$ distinct vertices.  This results in an arbitrary spanning tree being returned.
    \item The second source of error arises from using approximate probabilities rather than true probabilities to generate midpoints.
    \item Finally, the third source of error arises from using approximate matching sampling rather than exact matching sampling and approximate probabilities in the matching algorithm.
\end{enumerate}

\subsection{Solving Problem 1}
The clear solution to the first problem is to convert the algorithm from a Monte Carlo to Las Vegas approach.  In other words, if the walk generated during a non-final phase visits fewer than $\rho$ distinct vertices, we double the target length $\ell$, sample a new end vertex, and continue the walk generation.  In particular, we sample the new end vertex in the same way as before but use the current end vertex as the starting point.   We repeat this indefinitely until we visit enough distinct vertices.  Note that we can choose the initial target walk length to be high enough, $\Tilde{O}(n^3)$, such that we never need to generate more than one walk in a phase with high probability.  The walk that is returned is now guaranteed to contain exactly $\rho$ distinct vertices.

\subsection{Solving Problem 2}
For now, ignore the issue of sampling matchings.  We will show that we can completely remove the error of midpoint generation.  The same logic can be applied to the shortcut graph as well as sampling the walk endpoint.  These ideas are credited to Propp, see \cite{madryOther,madryThesis}.

Suppose that so far, the endpoint and every midpoint have been sampled and placed exactly.  We are now interested in sampling a midpoint $m$ between a start end pair $(p,q)$.  In particular, for some transition matrix $\mathcal{W}$, we want to sample $m$ from the distribution $P(m) = \frac{\mathcal{W}[p,m]\mathcal{W}[m,q]}{\mathcal{W}^2[p,q]}$.  We will assume that $\mathcal{W}^2[p,q] \geq \frac{1}{n^c}$, for some fixed $c>0$.  In fact, we will choose $c$ such that $\mathcal{W}^2[p,q] \leq \frac{2}{n^c}$ for any midpoint in the entire algorithm with probability at most $\frac{1}{n}$.  This choice of $c$ exists since $(p,q)$ appear sequentially in one partial walk.  Thus an event with probability $\mathcal{W}^2[p,q]$ has already occurred and we are only sampling a polynomial number of midpoints.  We can verify that $\mathcal{W}^2[p,q] \geq \frac{1}{n^c}$ by using our approximation of $\mathcal{W}^2$ with subtractive error of $\frac{1}{n^c}$ and checking if the approximation is greater than $\frac{1}{n^c}$.  In the very rare case where verification fails, we resort to centrally collecting the entire network at $\sfM$ and completing the algorithm by brute force.  Otherwise we can use Lemma \ref{MidpointTVDLemma} to approximate the distribution with total variation distance at most $\delta/2$ by choosing $\beta = \frac{\delta}{4 n^{c+1}}$.  We can transform this into an approximation with subtractive error $\delta$ by subtracting $\frac{\delta}{2}$ from each entry and then setting each negative entry to be zero.

Let the approximate probabilities with subtractive error be $\hat{P}$.  Let the sum of all of the approximate probabilties be $Z$.  Note that $Z \geq 1-n\delta$.  Finally we draw a random number $r$.  If $r \leq Z$, we draw the midpoint from the unnormalized distribution implied by $\hat{P}$.  Otherwise, we fail and again we simply collect the entire network by brute force.  Then we pick the midpoint conditioned on $r > Z$ and solve the remaining algorithm locally.  To avoid adding to the expected runtime we need to avoid ever collecting the graph with high probability, specifically we need this probability to be $O(\frac{1}{n})$.  Note that it is sufficient, since we sample fewer than $n^4$ midpoints with high probability, that $n\delta \leq \frac{1}{n^5}$, so we choose $\delta = \frac{1}{n^6}$.

\subsection{Solving Problem 3}
Unfortunately, we do not know how to eliminate the error associated with sampling perfect matchings.  Even if we could sample perfect matchings exactly, for example by using exponential time on machine $\sfM$, we still have the problem that $\sfM$ only has approximate probabilities for each transition matrix entry.  Since the set of all perfect matchings is exponential in size, we cannot use the trick we used to sample midpoints exactly.

As discussed earlier, if each machine representing  start-end pair $(p,q)$ could send its sequence $\Pi_{p,q}$ to $\sfM$ we could sample random walks exactly.  We strike a compromise between $\sfM$ receiving only a multiset of midpoints and $\sfM$ receiving each entire sequence $\Pi_{p,q}$.  The key insight is that each permutation of $\Pi_{p,q}$ has the same probability, assuming the absolute final midpoint remains the same.  This is because all midpoints in a given sequence $\Pi_{p,q}$ were sampled between the same start-end pair $(p,q)$.  Thus we can remove any error from sampling matchings by sending a compressed version of $\Pi_{p,q}$, its corresponding multiset, to $\sfM$.  As before we also need to separately send the chronologically final midpoint.  Then $\sfM$ can resample each sequence by choosing a uniformly random permutation of each multiset.

Note that we need $\Theta(\sqrt{n})$ words to represent each multiset, since that is how many distinct vertices the walk visits.  $\sfM$ does not have the bandwidth to receive $\sqrt{n}$ words from all $n$ machines.  This necessitates reducing the number of distinct vertices visited per phase.  In particular, we choose to visit $\sqrt[3]{n}$ distinct vertices per phase.  With this choice, there are only $n^{2/3}$ machines representing start-end pairs.  Thus $\sfM$ needs to receive a total of $\sqrt[3]{n}$ words from $n^{2/3}$ machines for a total of $n$ words.  This can be achieved in $O(1)$ rounds.  Reducing the number of distinct vertices visited per phase increases the total runtime to $\Tilde{O}(n^{2/3+\alpha})$.
\end{document}